\documentclass[twoside,bezier,12pt, reqno]{amsart}
\usepackage{amssymb,latexsym,epsfig,mathrsfs,graphpap,graphics,amsmath,amssymb}
\usepackage{amstext}
\usepackage{amsthm}
\usepackage[utf8]{inputenc}
\usepackage[english]{babel}
\usepackage{helvet}
\usepackage{courier}
\usepackage{tikz}
\usepackage{enumerate}
\usepackage{graphicx}
\usepackage{float}
\usepackage{subfigure}
\usepackage{framed}
\usetikzlibrary{shapes,arrows}

\newtheorem{theorem}{Theorem}[section]
\newtheorem{lemma}{Lemma}[section]

\newtheorem{proposition}{Proposition}[section]
\newtheorem{corollary}{Corollary}[section]

\theoremstyle{Definition}
\newtheorem{definition}{Definition}[section]
\newtheorem{example}{Example}[section]

\theoremstyle{remark}
\newtheorem{remark}[theorem]{Remark}
\numberwithin{equation}{section}

\begin{document}

\begin{flushleft}
 {\bf\Large  Special Affine Stockwell Transform: Theory, Uncertainty Principles and Applications }

\parindent=0mm \vspace{.2in}

{\bf{  Aamir H. Dar$^{1},$  M. Younus Bhat* $^{2}$ }}
\end{flushleft}

{{\it $^{1}$ Department of  Mathematical Sciences,  Islamic University of Science and Technology, Kashmir, India. E-mail: $\text{ahdkul740@gmail.com}$}}

{{\it $^{2}$ Department of  Mathematical Sciences,  Islamic University of Science and Technology, Kashmir, India. E-mail: $\text{gyounusg@gmail.com}$}}

\begin{quotation}
\noindent
{\footnotesize {\sc Abstract.}

 In this paper, we study the convolution structure in the special affine Fourier transform domain to
combine the advantages of the well known special affine Fourier and Stockwell transforms into a novel
integral transform coined as special affine Stockwell transform and investigate the associated constant Q-property
in the joint time-frequency domain. The preliminary analysis encompasses the derivation of
the fundamental properties, Rayleigh’s energy theorem,
inversion formula and range theorem. Besides, we also  derive a
direct relationship between the recently introduced special affine scaled Wigner distribution and the
proposed SAST. Further,  we establish  Heisenberg’s uncertainty principle, logarithmic uncertainty principle and Nazarov's uncertainty principle  associated with the proposed SAST. Towards the culmination
of this paper, some potential applications with simulation are presented. \\

{ Keywords:}Stockwell transform; Special affine Fourier transform; Special affine scaled·Wigner distribution; Time-frequency analysis; Uncertainty principle.\\
\noindent
\textit{2000 Mathematics subject classification:42C40; 81S30; 65R10; 44A35;42C15. } }
\end{quotation}

\section{Introduction}\label{sec intro}
\noindent
The special affine Fourier transform (SAFT) is a neoteric addition to the theory of Fourier transforms introduced by Abe and Sheridan \cite{wz1,wz2} and it is studied by
various researchers \cite{8s,9s}. The special affine Fourier transform is a six parameter time-frequency analysis tool that generalizes several existing time-frequency analysis tool including  the Fourier transform, the fractional Fourier trans-
form, the linear canonical transform (LCT) and the Fresnel transform \cite{3s,12s,13s}. The special affine Fourier transform or the the inhomogeneous canonical transform \cite{1s} can be regarded as a time-shifted and frequency-modulated version of well known linear canonical integral transforms \cite{100s,101s,102s}. Due to the time shifting and frequency-modulating parameters, the special affine Fourier transform is much more flexible than classical linear canonical transform and has
found wide applications  in digital signal processing and optical systems. The special affine Fourier transform of any signal $f(t)\in L^2(\mathbb R)$  is defined as \cite{1102s}
 \begin{equation}\label{qpft}
\mathbf S_{{\mathbf N}}[f(t)](w)=\int_{\mathbb R}f(t)\mathcal K_{{\mathbf N}}(t,w)dt,\,B\ne0
\end{equation}
where ${\mathbf N}=\begin{bmatrix}
A & B ;& p\\
C & D ;& q
\end{bmatrix}$ with $AD-BC=1$ is a augmented matrix parameter
and $\mathcal K_{{\mathbf N}}(t,w)$ denotes the
kernel of the special affine Fourier transform given by
\begin{equation}\label{ker qpft}
 K_{{\mathbf N}}(t,w)=K_Be^{\frac{i}{2B}\left(At^2+2t(p-w)-2w(Dp-Bq)+D(w^2+p^2)\right)}.
\end{equation}
where $K_B=\frac{1}{\sqrt{i 2\pi B}}.$
We only consider the case of $B\ne0$, since the SAFT is just a chirp multiplication
operation if $B=0$.The inversion formulae corresponding to SAFT is given by
\begin{equation}\label{2}
f(t)=\int_{\mathbb R}\mathbf S_{{{\mathbf N}}}[f](w){\mathcal K_{{{\mathbf N}}^{-1}}(w,t)}dw,
\end{equation}
where $${{\mathbf N}}^{-1}=\begin{bmatrix}
D & -B & Bq-Dp\\
-C & A & Cp-Aq
\end{bmatrix}$$
The Parseval’s formula for the SAFT reads as follows
\begin{equation}\label{3}
\langle f,g\rangle=\langle \mathbf S_{{{\mathbf N}}}[f],\mathbf S_{{{\mathbf N}}}[g]\rangle,\qquad\forall f,g\in L^2(\mathbb R).\\
\end{equation}
Since the inception of  the special affine Fourier transform, it has achieved  a  respectable status within a short span of time and has been applied to electrical, optical and communication systems,
and several other fields of science and technology \cite{15w,20w,21w}.
Apart from applications, the theoretical
framework of special affine Fourier transform has also been extensively studied and investigated including the Poisson summation formulae, sampling
theory, uncertainty principles, convolution theorems and so on \cite{15s}-\cite{19s}. Besides a lot of advantages, the special affine Fourier transform suffers a major drawback, as it is incapable of obtaining information about the local properties of the non-stationary signal owing to its global kernel. Thus, SAST is inadequate  in situations demanding a joint analysis of time and spectral
characteristics of a signal.\\
In signal processing the Fourier transform can be considered as one of the efficient  time-frequency representation tool. Gabor \cite{10g} in sixties introduced  various time-frequency analysis tools including, the short-time
Fourier transform, the Wigner distribution  or also the continuous wavelet transform
where all of these representations have a same common point, that is to analyse the signal  in both time and frequency domain simultaneously.\\
In spite of the fact that the short-time Fourier transform (STFT)\cite{st1} did much to ameliorate the limitations of FT, but still in some cases the STFT cannot track the signal dynamics properly for a
signal with both very high frequencies of short duration and very low frequencies of long duration, because of the fixed width of the
analysing window. Removing of the rigidity
of the window function is one of the motivations for continuous wavelet transform (WT) \cite{st2}. WTs have fascinated the researchers  with their versatile applicability and have been applied in a number of fields including sampling theory,   geophysics,  signal and
image processing, , quantum mechanics, astrophysics,  differential equations, neuroscience, medicine,chemistry and so on \cite{st3,st4,st5,st6}.
Although the WT have its own merits, it suffers from lack of phase information, poor directionality, shift sensitivity and other issues, which lead to unsatisfactory time-frequency distribution of non-stationary signals. To address the limitations of the wavelet transform, Stockwell et al. \cite{st11} proposes a new transform called as the Stockwell transform (ST)  by combining the merits of STFT and WT. In the context of time-frequency analysis classical ST can be viewed as multi-scale local FT or phase corrected WT and thus, provides more accurate information of local behaviour of a signal \cite{st12}.  Since its inception, ST has been successfully used to analyse signals in numerous applications, such as ground vibrations, biomedical imaging,
gravitational waves, geophysics, oceanology, optics,  and bio-informatics,
in general \cite{st13}-\cite{st22}. The
mathematical theory of the Stockwell transform is under development in different directions.

A  number of extensions of the ST have  have been proposed in recent years by coupling the Stockwell transform with some other appropriate transforms
such as, Fractional Fourier transform , linear canonical transform, Dunkl transform and others (see \cite{st30,st31,st32,st31a}). These extensions i.e. the fractional Stockwell transform (FRST) and linear canonical Stockwell transform (LCST) and Dunkl–Stockwell transform, analyse signals with abilities of multi-resolution, multi-angle , multi-scale and temporal localization and hence are very much suitable in dealing with chirp-like signals which is ubiquitous in nature. Keeping
in view the extensions of STs along with the profound applicability of
the SAFT, we are deeply inspired to intertwine their respective merits in a novel transform namely the
 special affine Stockwell transform (SAST) by employing the convolution structure of special affine
Fourier transforms. The prime advantage of this intertwining lies in
the fact that the proposed  SAST inherits the excellent mathematical
properties of both the traditional ST and SAST along with some fascinating properties of its own.  To be specific, it has the capability to display time and SAST domain information jointly. Moreover, many old and new time-frequency transforms such as the traditional ST,
the FRST,the LCST and the Fresnel-Stockwell transform can be derived as particular cases
of the proposed SAST. The proposed transform  enjoys certain extra degrees of freedom and thus
provide higher flexibility in applications to various
disciplines of science and engineering including non-stationary signal processing, optics, harmonic analysis, geophysics, quantum mechanics and operator theory.
\subsection{Paper Contributions}\,\\

The contributions  of this  paper are summarized below:\\

\begin{itemize}

\item  We introduce the novel integral transform coined as the  special affine Stockwell transform (SAST).\\

\item To study the fundamental properties including  Rayleigh’s energy theorem, inversion formula and  range theorem.\\

  \item  To establish a direct relationship between the special affine scaled Wigner  distribution and
the proposed SAST.\\

  \item To extend the scope of the study, we formulate Heisenberg's uncertainty principle, the logarithmic uncertainty principle and Nazarov's uncertainty principle associated with the SAST.\\
\item  To testify the efficiency Of the proposed SAST, some potential applications with simulation are presented. .

\end{itemize}
\subsection{Paper Outlines}\,\\
The paper is organized as follows: Section \ref{sec 1} is devoted for an overview of the pre-requisites
including the  ST and convolution structure of SAFT. In Section \ref{sec 2}, we introduce the notion of   special affine Stockwell transform (SAST) and investigate  its fundamental properties. Subsequently, a direct
 relationship between the recently introduced special affine scaled Wigner distribution and the proposed SAST is established in Section \ref{secrel}. In Section \ref{secucp} we establish the Heisenberg's uncertainty principle,
  the logarithmic uncertainty principle and the Nazarov's uncertainty principle  associated with the SAST. Finally, the potential applicability with simulations is briefly discussed
in Section \ref{secpa}.

\section{Preliminary}\label{sec 1}
In this section,  we recall the convolution structure of the special affine Fourier transform and give a brief  review to the the Stockwell transform and existing results required to follow the remaining content of the paper.\\
For notational convenience, we shall write the augmented matrix parameter ${\mathbf N}=\begin{bmatrix}
A & B ;& p\\
C & D ;& q
\end{bmatrix}$ as ${\mathbf N}=(A,B,C,D;p,q)$, in the rest of paper.
\begin{definition}\cite{15s}
For any pair of functions $f,g\in L^2(\mathbb R)$ the special affine convolution $\ast_{{\mathbf N}}$
with respect to the matrix  parameter $\mathbf N=$ is defined by
\begin{equation}\label{eqn qp conv}
(f\ast_{{\mathbf N}} g)(b)=K_B\int_{\mathbb R} f(t)g(b-t)e^{\frac{i}{2B}\left(Dp^2-2At(b-t)\right)}dt
\end{equation}
\end{definition}
Under this convolution, convolution theorem  in the SAFT domain reads
\begin{equation}\label{eqn qpft conv}
\mathbf S_{{\mathbf N}}[f\ast_{{\mathbf N}} g](w)=e^{\frac{i}{2b}\left(2w(Dp-Bq)-dw^2\right)}\mathbf S_{{\mathbf N}}[f](w)\mathbf S_{{\mathbf N}}[g](w)
\end{equation}

   Prior to introduction of ordinary Stockwell transform and its basic
properties, we shall recall three elementary operators dilation, modulation and translation in $L^2(\mathbb R)$ as  ${\mathcal D}_a,$ ${\mathcal M}_a$ and $\mathbf T_b$ where $a\in L^2(\mathbb R^+)$ and $b\in L^2(\mathbb R),$ respectively
\begin{equation*}
({\mathcal D}_a\Psi)(t)=|a|\Psi(at),\,({\mathcal M}_a\Psi)(t)=e^{iat}\Psi(t),\,({\mathbf T}_b\Psi)(t)=\Psi(t-b).
\end{equation*}
Using above elementary operators on  $\Psi\in L^2(\mathbb R)$, we define a family of analyzing functions as
\begin{equation}\label{fun psi}
\Psi_{a,b}(t)= {\mathcal M}_a{\mathbf T}_b {\mathcal D}_a\Psi(t)=|a|e^{iat}\Psi(a(t-b));\ a\in \mathbb R^+,\,b\in \mathbb R.
\end{equation}
With the help of above mentioned operators, Stocwell transform can be defined as
\begin{definition}[Stockwell transform \cite{st11}]For a given window function $\Psi\in L^2(\mathbb R),$ the Stockwell transform (ST) of a signal $f(t)\in L^2(\mathbb R)$ w.r.t $\Psi$ is defined as
\begin{eqnarray}
\label{st 1}[\mathcal S_{\Psi}f](a,b)&=&\frac{1}{\sqrt{2\pi}}\int_{\mathbb R}f(t)\overline{\Psi_{a,b}(t)}dt\\
\label{st2}&=&\frac{1}{\sqrt{2\pi}}\int_{\mathbb R}f(t)ae^{-iat}\overline{\Psi(a(t-b))}dt\\
\label{st3}&=&\frac{1}{\sqrt{2\pi}}\left([f(t)e^{-iat}]\ast[a\tilde\Psi(at)]\right)(b)\\
\label{st4}&=&\frac{1}{\sqrt{2\pi}}\left([{\mathcal M}_{-a}f]\ast[{\mathcal D}_a\tilde\Psi]\right)(b),
\end{eqnarray}
\end{definition}
where $\tilde\Psi(t)=\overline{\Psi}(-t)$ denotes the involution and $\ast$ represents the classical convolution in in the Fourier domain.\\
 The orthogonality relation and inversion formula   corresponding to Stockwell transform are given by
\begin{equation}\label{st ortho}
\left\langle \mathcal S_{\Psi}f,\mathcal S_{\Psi}g\right\rangle=\mathcal C_\Psi\left\langle f,g\right\rangle  ,\,\forall f,g \in L^2(\mathbb R),
\end{equation}
and
\begin{equation}\label{st inver}
f(t)=\frac{1}{\mathcal C_\Psi}\int_{\mathbb R}\int_{\mathbb R^+}[\mathcal S_{\Psi}f](a,b)\Psi_{a,b}(t)dadb,\, a.e \quad t\in \mathbb R.
\end{equation}
where $\mathcal C_\Psi$ denotes the admissibility condition.
\section{ Special Affine Stockwell transform and the Associated Constant Q-Property}\label{sec 2}
In this section, we shall formally introduce a novel transform in the context of time-frequency analysis
namely, special affine Stockwell transform (SAST) by intertwining the advantages of the SAFT and the Stockwell transform. The proposed SAST offers a joint information of time and special affine spectrum
in the time-SAFT domain.\\
\subsection{Definition of the SAST}\,\\
Motivated by the convolution theorem of special affine Fourier transform domain and by the basic properties of Stockwell transform, the  special affine Stockwell transform SAST of any signal $f(t)\in L^2(\mathbb R)$ w.r.t a given window function $\Psi(t)$ is defined as
\begin{eqnarray}
\nonumber[SAST^{{\mathbf N}}_\Psi f(t)](a,b)&=&\frac{1}{\sqrt{2\pi}}\left[({\mathcal M}_{-a}f)\ast_{{\mathbf N}}({\mathcal D}_a\tilde\Psi) \right](b)\\
\nonumber&=&\frac{1}{\sqrt{2\pi}}\left[e^{-iat}f(t)\ast_{{\mathbf N}} a\tilde\Psi(at)\right](b)\\
\nonumber&=&\frac{1}{\sqrt{2\pi}}K_Be^{i\frac{D}{2B}p^2}\big\langle e^{-ia(.)}f(.)e^{-i\frac{A}{B}(.)(b-.)},\Psi_{a,b}(.)\big\rangle\\
\nonumber&=&\frac{1}{\sqrt{2\pi}}\left\langle f(t),{\Psi_{\mathbf N,a,b}(t)}\right\rangle\\
\label{eqn qpst}&=&\frac{1}{\sqrt{2\pi}}\int_{\mathbb R} f(t)\overline{\Psi_{\mathbf N,a,b}(t)}dt,
\end{eqnarray}
where the family of functions $\Psi_{{{\mathbf N}},a,b}(t)$ satisfies
\begin{equation}\label{fun psi mu}
\Psi_{{{\mathbf N}},a,b}(t)=|a|K^*_B\Psi(a(t-b))e^{iat+\frac{i}{2B}(2At(b-t)-Dp^2)}.\\
\end{equation}
With the virtue of above equations we have following definition
\begin{definition}\label{def qpst} For a given matrix parameter ${{\mathbf N}}=(A,B,C,D;p,q),$ and a given  window function $\Psi\in L^2(\mathbb R),$  the special affine Stockwell transform of any signal $f\in L^2(\mathbb R)$ is defined by
\begin{equation}\label{eqn def qpst}
[SAST^{{\mathbf N}}_\Psi f(t)](a,b)=\frac{1}{\sqrt{2\pi}}\int_{\mathbb R} f(t)\overline{\Psi_{{{\mathbf N}},a,b}(t)}dt,
\end{equation}
where $\Psi_{{{\mathbf N}},a,b}(t)$ is given by (\ref{fun psi mu}).
\end{definition}

Equation, (\ref{eqn def qpst}) takes form
\begin{equation}\label{eqn qpst2}
[SAST^{{\mathbf N}}_\Psi f(t)](a,b)=\frac{|a|}{\sqrt{2\pi}}K_B\int_{\mathbb R} f(t)\overline{\Psi(a(t-b))}e^{-iat-\frac{i}{2B}(2At(b-t)-Dp^2)}dt.
\end{equation}
Further, (\ref{eqn qpst2}) can be rewritten as
\begin{equation}\label{eqn qpst2a}
[SAST^{{\mathbf N}}_\Psi f(t)](a,b)=
\frac{|a|}{\sqrt{2\pi}}K_Be^{\frac{i}{2B}Dp^2}\int_{\mathbb R} f(t)e^{\frac{-i}{B}At(b-t)}\overline{\Psi(a(t-b))}e^{-iat}dt.
\end{equation}
From (\ref{eqn qpst2a}), we observe
that the computation of the special affine Stockwell transform performs the following
three actions in unison:\\
\begin{itemize}
\item A product by a chirp signal, i.e., $f(t)\rightarrow \tilde f(t)=f(t)e^{-i\frac{A}{B}t(b-t)}.$\\

\item The Stockwell transform, i.e., $\tilde f(t)\rightarrow [\mathcal S_\Psi\tilde f](a,b).$\\
\item  Another multiplication by chirp signal, i.e, $[\mathcal S_\Psi\tilde f](a,b)=[SAST^{{\mathbf N}}_\Psi f(t)](a,b)=K_Be^{\frac{i}{2B}Dp^2}[\mathcal S_\Psi\tilde f](a,b).$
\end{itemize}
 Thus the computational load of the quadratic-phase Stockwell transform is essentially
dictated by the classical Stockwell transform. Moreover, the aforementioned scheme is depicted in Figure \ref{fig 1}.\\

%

As a consequence of Definition \ref{def qpst}, we have the following deductions

\begin{itemize}
\item For the  matrix parameter ${{\mathbf N}}=(A,B,C,D;0,0),$ Definition \ref{def qpst} boils down to the linear canonical Stockwell transform \cite{st30}
   \begin{equation*}
[LCST^{{\mathbf N}}_\Psi f(t)](a,b)=\frac{|a|}{\sqrt{2\pi}}\int_{\mathbb R} f(t)\overline{\Psi(a(t-b))}e^{-it\left(a+(b-t)\frac{A}{B}\right)}dt.
\end{equation*}
\item  For the matrix parameter ${{\mathbf N}}=(\cos\theta,\sin\theta,-\sin\theta,\cos\theta;0,0),\quad \theta\ne n\pi$, we can obtain a  novel fractional Stockwell transform
   \begin{equation*}
[FrST^{{\mathbf N}}_\Psi f(t)](a,b)=\frac{|a|}{\sqrt{2\pi}}\int_{\mathbb R} f(t)\overline{\Psi(a(t-b))}e^{-it\left(a+(b-t)\cot\theta\right)}dt.
\end{equation*}
 \item For the matrix parameter ${{\mathbf N}}=(1,B,0,1;0,0),\,B\ne 0,$ we can obtain a  novel Fresnel-Stockwell
transform given by
 \begin{equation*}
[SAST^{{\mathbf N}}_\Psi f(t)](a,b)=\frac{|a|}{\sqrt{2\pi}}\int_{\mathbb R} f(t)\overline{\Psi(a(t-b))}e^{-it\left(a+(b-t)\right)}dt.
\end{equation*}
\item For the matrix ${{\mathbf N}}=(0,1,-1,0;0,0),$ Definition \ref{def qpst} reduces to the classical Stockwell transform (\ref{st2}).\\
 \end{itemize}

\begin{lemma}\label{lem 1}Let  $\Psi\in L^2(\mathbb R)$ be a window function, then special affine Stockwell transform of $\Psi_{{{\mathbf N}},a,b}$ is given by
\begin{eqnarray}\label{qpft of psi}
\nonumber\mathbf S_{{\mathbf N}}\left[\Psi_{{{\mathbf N}},a,b}\right](w)
&=&K^*_BK^{-1}_B{e^{iab}}\mathcal K_{{\mathbf N}}(b,w)e^{\frac{i}{2B}\left(2\frac{w}{a}(Dp-Bq)-D\left(\frac{w^2}{a^2}-2p^2\right)\right)}\\
\nonumber&&\qquad\qquad\qquad\times\mathbf S_{{\mathbf N}}\left[ e^{i\left\{y\left(1-\left(1-\frac{1}{a}\right)\frac{p}{B}\right)-\frac{Ay^2}{2B}\left(1+\frac{1}{a^2}\right)\right\}}\Psi(y)\right]\left(\frac{w}{a}\right),\\
\end{eqnarray}
where $\Psi_{{{\mathbf N}},a,b}(t)$ is given by (\ref{fun psi mu}).
\end{lemma}
\begin{proof}
From (\ref{qpft}), we have
\begin{eqnarray*}
\nonumber&&\mathbf S_{{\mathbf N}}\left[\Psi_{{{\mathbf N}},a,b}\right](w)\\
\nonumber&&\qquad={|a|}K_BK^*_B\int_{\mathbb R} \Psi(a(t-b))e^{iat+\frac{i}{2B}(2At(b-t)-Dp^2)}e^{\frac{i}{2B}(At^2+2t(p-w)-2w(Dp-Bq)+D(w^2+p^2))}dt\\
\nonumber&&\qquad=K_BK^*_B\int_{\mathbb R} \Psi(y)e^{\left\{ia\left(b+\frac{y}{a}\right)+\frac{i}{2B}\left[2A\left(b+\frac{y}{a}\right)\left(\frac{-y}{a}\right)-Dp^2\right]\right\}}\\
\nonumber&&\qquad\qquad\qquad\qquad\times e^{\frac{i}{2B}\left\{A\left(b+\frac{y}{a}\right)^2+2\left(b+\frac{y}{a}\right)(p-w)-2w(Dp-Bq)+D(w^2+p^2)\right\}}dy\\
\nonumber&&\qquad={e^{iab}}K_BK^*_B\int_{\mathbb R} \Psi(y)e^{\left\{iy-i\frac{A}{2B}\frac{y^2}{a^2}\right\}}e^{\frac{i}{2B}\left\{Ab^2+2b(p-w)+2\frac{y}{a}(p-w)-2w(Dp-Bq)+Dw^2\right\}}dy\\
\nonumber&&\qquad={e^{iab}}K_Be^{\frac{i}{2B}\left(Ab^2+2b(p-w)-2w(Dp-Bq)+D(w^2+p^2)\right)}K^*_B\int_{\mathbb R} \Psi(y)e^{\left\{iy-i\frac{A}{2B}\frac{y^2}{a^2}\right\}}\\
\nonumber&&\qquad\qquad\qquad\qquad\times e^{\frac{i}{2B}\left\{2\frac{y}{a}(p-w)-Dp^2\right\}}dy\\
\nonumber&&\qquad={e^{iab}}\mathcal K_{{\mathbf N}}(b,w)K^*_B\int_{\mathbb R} \Psi(y)e^{\left\{iy-i\frac{A}{2B}\frac{y^2}{a^2}\right\}}e^{\frac{i}{2B}\left\{Ay^2+2y(p-\frac{w}{a})-2\frac{w}{a}(Dp-Bq)+D(\frac{w^2}{a^2}+p^2)\right\}}\\
\nonumber&&\qquad\qquad\qquad\qquad\times e^{\frac{i}{2B}\left\{2\frac{y}{a}(p-w)-Dp^2-Ay^2-2y(p-\frac{w}{a})+2\frac{w}{a}(Dp-Bq)-D(\frac{w^2}{a^2}+p^2)\right\}}dy\\
\nonumber&&\qquad=K^*_BK^{-1}_B{e^{iab}}\mathcal K_{{\mathbf N}}(b,w)e^{\frac{i}{2B}\left(2\frac{w}{a}(Dp-Bq)-D\left(\frac{w^2}{a^2}-2p^2\right)\right)}\\
\nonumber&&\qquad\qquad\qquad\qquad\times\int_{\mathbb R} \Psi(y)e^{i\left\{y\left(1-\left(1-\frac{1}{a}\right)\frac{p}{B}\right)-\frac{Ay^2}{2B}\left(1+\frac{1}{a^2}\right)\right\}}\mathcal K_{{\mathbf N}}\left(y,\frac{w}{a}\right),
 dy\\
 \nonumber&&\qquad=K^*_BK^{-1}_B{e^{iab}}\mathcal K_{{\mathbf N}}(b,w)e^{\frac{i}{2B}\left(2\frac{w}{a}(Dp-Bq)-D\left(\frac{w^2}{a^2}-2p^2\right)\right)}\mathbf S_{{\mathbf N}}\left[ \Psi(y)e^{i\left\{y\left(1-\left(1-\frac{1}{a}\right)\frac{p}{B}\right)-\frac{Ay^2}{2B}\left(1+\frac{1}{a^2}\right)\right\}}\right]\left(\frac{w}{a}\right),
\end{eqnarray*}
which completes proof.\\
\end{proof}

Next, we shall present some illustrative examples for the demonstration of the SAST (\ref{eqn def qpst}). For this we shall compute the SAST of some functions with respect to the Gaussian window function  $\Psi(t)=e^{-\pi t^2}$.
\begin{example}\label{exa}
(a)Consider the Gaussian function $f(t)=e^{-\sigma t^2},\,K>0$, then the  SAST of $f(t)$ is  given by
\begin{eqnarray}
\nonumber[SAST^{{\mathbf N}}_\Psi f(t)](a,b)&=&\frac{|a|}{\sqrt{2\pi}}K_B\int_{\mathbb R} f(t)\overline{\Psi(a(t-b))}e^{-iat-\frac{i}{2B}(2At(b-t)-Dp^2)}dt\\
\nonumber&=&\frac{|a|}{\sqrt{2\pi}}K_B\int_{\mathbb R}e^{-\sigma t^2} e^{-\pi a^2(t-b)^2}e^{-iat-\frac{i}{2B}(2At(b-t)-Dp^2)}dt\\
\nonumber&=&\frac{|a|e^{-\pi a^2b^2+i\frac{Dp^2}{2B}}}{\sqrt{2\pi}}K_B\int_{\mathbb R} e^{-(\pi a^2+\sigma-iA/B)t^2+(2\pi a^2 b-ia-iAb/B)t}dt\\
\label{exa11}&=&\frac{|a|e^{-\pi a^2b^2+i\frac{Dp^2}{2B}}}{2\sqrt{\pi A+i\sigma\pi B+iB\pi^2 a^2}}e^{\frac{B(2\pi a^2b-ia-iAb/B)^2}{4(\pi a^2B+\sigma B-iA)}}.
\end{eqnarray}
(b) For a  modulation   function  $f(t)=e^{iKt},\,K>0$, the SAST  is computed as
\begin{eqnarray}
\nonumber[SAST^{{\mathbf N}}_\Psi f(t)](a,b)&=&\frac{|a|}{\sqrt{2\pi}}K_B\int_{\mathbb R} f(t)\overline{\Psi(a(t-b))}e^{-iat-\frac{i}{2B}(2At(b-t)-Dp^2)}dt\\
\nonumber&=&\frac{|a|}{\sqrt{2\pi}}K_B\int_{\mathbb R}e^{iKt} e^{-\pi a^2(t-b)^2}e^{-iat-\frac{i}{2B}(2At(b-t)-Dp^2)}dt\\
\nonumber&=&\frac{|a|e^{-\pi a^2b^2+i\frac{Dp^2}{2B}}}{\sqrt{2\pi}}K_B\int_{\mathbb R} e^{-(\pi a^2-iA/B)t^2+(iK+2\pi a^2 b-ia-iAb/B)t}dt\\
\label{exa11}&=&\frac{|a|e^{-\pi a^2b^2+i\frac{Dp^2}{2B}}}{2\sqrt{\pi A+iB\pi^2 a^2}}e^{\frac{B(iK+2\pi a^2b-ia-iAb/B)^2}{4(\pi a^2B-iA)}}.
\end{eqnarray}
(c) For a constant function $f(t)=M,$, the SAST of $f(t)$  can be computed as
\begin{eqnarray}
\nonumber[SAST^{{\mathbf N}}_\Psi f(t)](a,b)&=&\frac{|a|}{\sqrt{2\pi}}K_B\int_{\mathbb R} M\overline{\Psi(a(t-b))}e^{-iat-\frac{i}{2B}(2At(b-t)-Dp^2)}dt\\
\label{exa11}&=&\frac{M|a|e^{-\pi a^2b^2+i\frac{Dp^2}{2B}}}{2\sqrt{\pi A+iB\pi^2 a^2}}e^{\frac{B(2\pi a^2b-ia-iAb/B)^2}{4(\pi a^2B-iA)}}.
\end{eqnarray}
\end{example}
For different choices of parameter ${{\mathbf N}}=(A,B,C,D,E),$ the corresponding
 special affine Stockwell transforms of  (\ref{exa11}) and (\ref{exa12}) are plotted in Figures \ref{fig 2}-\ref{fig 8}.

  In next subsection, our goal is to analyze the joint resolution of the special affine Stockwell
transform, in the time-special affine-frequency domain. In this direction, we present the  following Proposition.

\begin{proposition}\label{prop 1} Let $\mathbf S_{{\mathbf N}}[f]$ and $[SAST^{{\mathbf N}}_\Psi f(t)](a,b)$ denote the special affine Fourier transform and special affine Stockwell transform of any signal $f$ in $L^2(\mathbb R),$ respectively. Then we have
\begin{eqnarray}\label{prop 1 eqn}
\nonumber[SAST^{{\mathbf N}}_\Psi f(t)](a,b)&=&\frac{{e^{-iab}}K_B}{\sqrt{2\pi}K^*_B}\int_{\mathbb R}e^{\frac{-i}{2B}\left(2\frac{w}{a}(Dp-Bq)-D\left(\frac{w^2}{a^2}-2p^2 \right)\right)} \mathbf S_{{\mathbf N}}\left[f(t)\right](w)\\
\nonumber&&\qquad\qquad\qquad\times\overline{\mathbf S_{{\mathbf N}}\left[ e^{i\left\{y\left(1-\left(1-\frac{1}{a}\right)\frac{p}{B}\right)-\frac{Ay^2}{2B}\left(1+\frac{1}{a^2}\right)\right\}}\Psi(y)\right]}\left(\frac{w}{a}\right)\overline{\mathcal K_{{\mathbf N}}(b,w)}dw,\\
\end{eqnarray}
where ${K_{{\mathbf N}}(b,w)}$ is given by (\ref{ker qpft}).
\end{proposition}
\begin{proof}

By the virtue of (\ref{eqn qpst}) and  Parseval identity of SAFT, we have
\begin{eqnarray*}
\nonumber[SAST^{{\mathbf N}}_\Psi f(t)](a,b)&=&\frac{1}{\sqrt{2\pi}}\left\langle f(t),{\Psi_{{{\mathbf N}},a,b}(t)}\right\rangle\\
\nonumber&=&\frac{1}{\sqrt{2\pi}}\left\langle \mathbf S_{{\mathbf N}}\left[f(t)\right](w), \mathbf S_{{\mathbf N}}\left[\Psi_{{{\mathbf N}},a,b}(t)\right](w)\right\rangle\\
\nonumber&=&\frac{1}{\sqrt{2\pi}}\int_{\mathbb R} \mathbf S_{{\mathbf N}}\left[f(t)\right](w)\overline{ \mathbf S_{{\mathbf N}}\left[\Psi_{{{\mathbf N}},a,b}(t)\right]}(w)dw.
\end{eqnarray*}

Using Lemma \ref{lem 1}, above equation yields
\begin{eqnarray*}
[SAST^{{\mathbf N}}_\Psi f(t)](a,b)&=&\frac{{e^{-iab}}K_B}{\sqrt{2\pi}K^*_B}\int_{\mathbb R}e^{\frac{-i}{2B}\left(2\frac{w}{a}(Dp-Bq)-D\left(\frac{w^2}{a^2}-2p^2 \right)\right)} \mathbf S_{{\mathbf N}}\left[f(t)\right](w)\\
&&\qquad\qquad\qquad\times\overline{\mathbf S_{{\mathbf N}}\left[ e^{i\left\{y\left(1-\left(1-\frac{1}{a}\right)\frac{p}{B}\right)-\frac{Ay^2}{2B}\left(1+\frac{1}{a^2}\right)\right\}}\Psi(y)\right]}\left(\frac{w}{a}\right)\overline{\mathcal K_{{\mathbf N}}(b,w)}dw.\\
\end{eqnarray*}
Which completes the proof.\\
\end{proof}

From (\ref{eqn qpst}) and (\ref{prop 1 eqn}), we  conclude that if the analyzing functions $\Psi_{{{\mathbf N}},a,b}(t)$  are supported in the time domain or SAFT domain , then the inner product $\left\langle f(t),\Psi_{{{\mathbf N}},a,b}(t)\right\rangle$ ensures that the proposed transform $[SAST^{{\mathbf N}} f(t)](a,b)$ is accordingly supported in the respective domains. This implies that the special affine Stockwell transform is capable of providing the simultaneous information of the
time and the special affine frequency in the time-frequency domain. Assuming that  $\Psi(t)$ is the window with with finite centre $\mathbf E_\Psi$ and finite radius $\Delta_\Psi$ in the time domain. Then the centre and radii of the time-domain window function $\Psi_{{{\mathbf N}},a,b}(t)$ of the proposed SAST is given by
\begin{equation}\label{tfd 1}
\mathbf{E}\left[\Psi_{{{\mathbf N}},a,b}(t)\right]=\dfrac{\displaystyle\int_{\mathbb R}t|\Psi_{{{\mathbf N}},a,b}(t)|^2dt}{\displaystyle\int_{\mathbb R}|\Psi_{{{\mathbf N}},a,b}(t)|^2dt}=\dfrac{\displaystyle\int_{\mathbb R}t|\Psi_{a,b}(t)|^2dt}{\displaystyle\int_{\mathbb R}|\Psi_{a,b}(t)|^2dt}=\mathbf{E}\left[\Psi_{a,b}(t)\right]=b+\frac{\mathbf{E}_\Psi}{a}\\
\end{equation}

and

\begin{eqnarray}
\nonumber\Delta\left[\Psi_{{{\mathbf N}},a,b}(t)\right]&=&\left(\dfrac{\displaystyle\int_{\mathbb R}\left(t-b-\frac{\mathbf{E}_\Psi}{a}\right)|\Psi_{{{\mathbf N}},a,b}(t)|^2dt}{\displaystyle\int_{\mathbb R}|\Psi_{{{\mathbf N}},a,b}(t)|^2dt}\right)^{\frac{1}{2}}\\
\nonumber&=&\left(\dfrac{\displaystyle\int_{\mathbb R}\left(t-b-\frac{\mathbf{E}_\Psi}{a}\right)|\Psi_{a,b}(t)|^2dt}{\displaystyle\int_{\mathbb R}|\Psi_{a,b}(t)|^2dt}\right)^{\frac{1}{2}}\\
\label{tfd 2}&=&\Delta\left[\Psi_{a,b}(t)\right]=\frac{\Delta_\Psi}{a}
\end{eqnarray}
Let $\mathbf{H}(w)$  be the window function in the SAFT domain given by
\begin{equation*}
\mathbf{H}(w)={\mathbf S_{{\mathbf N}}\left[ e^{i\left\{y\left(1-\left(1-\frac{1}{a}\right)\frac{p}{B}\right)-\frac{Ay^2}{2B}\left(1+\frac{1}{a^2}\right)\right\}}\Psi(y)\right]}\left({w}\right).
\end{equation*}
Then, we can find the center and radius of the SAFT domain window function
\begin{equation}\label{3.10}
\mathbf{H}\left(\frac{w}{a}\right)={\mathbf S_{{\mathbf N}}\left[ e^{i\left\{y\left(1-\left(1-\frac{1}{a}\right)\frac{p}{B}\right)-\frac{Ay^2}{2B}\left(1+\frac{1}{a^2}\right)\right\}}\Psi(y)\right]}\left(\frac{w}{a}\right)
\end{equation}

appearing in Eq.(\ref{prop 1 eqn}) as

\begin{equation}\label{3.11}
\mathbf{E}\left[\mathbf{H}\left(\frac{w}{a}\right)\right]= a\mathbf{E}_{\mathbf{H}},\qquad \Delta\left[\mathbf{H}\left(\frac{w}{a}\right)\right]= a\Delta_{\mathbf{H}}
\end{equation}
Hence,  the Q-factor(the ratio of width to center) of the special affine Fourier  domain function of the proposed transform is given by
\begin{equation}\label{3.12}
Q=\frac{\Delta\left[\mathbf{H}\left(\frac{w}{a}\right)\right]}{\mathbf{E}\left[\mathbf{H}\left(\frac{w}{a}\right)\right]}=\frac{\Delta_{\mathbf{H}}}{\mathbf{E}_{\mathbf{H}}}=constant
\end{equation}
Eq. (\ref{3.12}) illustrates that the window function is independent of both the of the matrix ${{\mathbf N}}=(A,B,C,D;p,q)$ and the scaling parameter $a$. This shows that the proposed SAST has  the constant-Q property. It should be noted here that the constant-Q transform plays a vital role in the  signal processing and are mostly meant for audio signals  as they require an analysis as well as a synthesis scheme \cite{st50,st51}.\\

Thus, the $[SAST^{{\mathbf N}}_{\Psi}f(t)](a,b)$ gives local information of the signal $f(t)$ with the time window as

\begin{equation}\label{3.13}
\left[b+\frac{\mathbf{E}_\Psi}{a}-\frac{\Delta_\Psi}{a},\,b+\frac{\mathbf{E}_\Psi}{a}+\frac{\Delta_\Psi}{a}\right].
\end{equation}
Similarly, we can infer from (\ref{3.11}) and Proposition \ref{prop 1} that the proposed transform $[SAST^{{\mathbf N}}_{\Psi}f(t)](a,b)$
 gives local information about the SAFT spectrum of the signal $f(t)$ in the window
\begin{equation}\label{3.14}
\left[a\mathbf{E}_{\mathbf{H}}-a{\Delta_{\mathbf{H}}},\, a\mathbf{E}_{\mathbf{H}}+a{\Delta_{\mathbf{H}}}\right].
\end{equation}

Hence, the the joint resolution of the proposed SAST $[SAST^{{\mathbf N}}_{\Psi}f(t)](a,b)$ in the time-frequency domain is described by a
flexible window having a total spread $4\Delta_{\mathbf{H}}\Delta_\Psi$ and
is given by
\begin{equation}\label{3.15}
\left[b+\frac{\mathbf{E}_\Psi}{a}-\frac{\Delta_\Psi}{a},\,b+\frac{\mathbf{E}_\Psi}{a}+\frac{\Delta_\Psi}{a}\right]\times\left[a\mathbf{E}_{\mathbf{H}}-a{\Delta_{\mathbf{H}}},\, a\mathbf{E}_{\mathbf{H}}+a{\Delta_{\mathbf{H}}}\right]=4\Delta_{\mathbf{H}}\Delta_\Psi.
\end{equation}
\subsection{Basic properties of the special affine Stockwell transform}
\begin{theorem}\label{thm p1}Let $\Psi,\Phi\in L^2(\mathbb R)$ be two window functions and $f,g$ be two functions in $L^2(\mathbb R)$. Then the proposed special affine Stockwell transform satisfies the
following properties:\\
(i)Linearity: $$\left[SAST^{{\mathbf N}}_\Psi(\alpha f+\beta g)\right](a,b)= \alpha\left[SAST^{{\mathbf N}}_\Psi f\right](a,b)+\beta\left[SAST^{{\mathbf N}}_\Psi g\right](a,b).$$\\
(ii)Anti-linearity: $$\left[SAST^{{\mathbf N}}_{\alpha\Psi+\beta\Phi}f\right](a,b)= \bar\alpha\left[SAST^{{\mathbf N}}_\Psi f\right](a,b)+\bar\beta\left[SAST^{{\mathbf N}}_\Phi f\right](a,b).$$\\
(iii)Translation: $$\left[SAST^{{\mathbf N}}_\Psi({\mathbf T}_k f)\right](a,b)=e^{-iak-\frac{i}{2B}\left(2Ak(b-k)\right)}\left[SAST^{{\mathbf N}}_\Psi({\mathcal M}_{\frac{Ak}{B}} f)\right](a,b-k).$$\\
(iv)Scaling: $$\left[SAST^{{\mathbf N}}_\Psi( f(\sigma t))\right](a,b)=\left[SAST^{\bar{{\mathbf N}}}_\Psi f\right]\left(\frac{a}{\sigma}, \sigma b\right),\quad \bar{{\mathbf N}}=(A,\sigma^2B,C,\sigma^2D;p,q)$$\\
(v)Parity:$$\left[SAST^{{\mathbf N}}_\Psi( f(- t))\right](a,b)=\left[SAST^{{\mathbf N}}_\Psi f\right](-a,-b).$$ \\
Where $\alpha,\beta\in \mathbb C,$  $k\in\mathbb R $ and $\sigma\in\mathbb R ^+.$
\end{theorem}
\begin{proof}For the sake of brevity we
We omit proof of (i) and (ii).\\
(iii) From (\ref{eqn qpst2}), we have for any real $k,$
\begin{eqnarray*}
&&\left[SAST^{{\mathbf N}}_\Psi({\mathbf T}_k f)\right](a,b)\\
&&=\frac{|a|}{\sqrt{2\pi}}K_B\int_{\mathbb R} f(t-k)\overline{\Psi(a(t-b))}e^{-iat-\frac{i}{2B}\left(2At(b-t)-Dp^2\right)}dt\\
&&=\frac{|a|}{\sqrt{2\pi}}K_B\int_{\mathbb R} f(y)\overline{\Psi(a(y-b+k))}e^{-ia(y+k)-\frac{i}{2B}\left(2A(y+k)(b-(y+k)-Dp^2\right)}dy\\
&&=\frac{|a|}{\sqrt{2\pi}}K_B\int_{\mathbb R} f(y)\overline{\Psi(a(y-(b-k)))}e^{-iay-\frac{i}{2B}\left(2Ay((b-k)-y)-Dp^2\right)}e^{-iak-\frac{i}{2B}\left(2Ak((b-k)-y)\right)}dy\\
&&=e^{-iak-\frac{i}{2B}\left(2Ak(b-k)\right)}\frac{|a|}{\sqrt{2\pi}}K_B\int_{\mathbb R} e^{\frac{Ak}{B}y}f(y)\overline{\Psi(a(y-(b-k)))}e^{-iay-\frac{i}{2B}\left(2Ay((b-k)-y)-Dp^2\right)}\\
&&=e^{-iak-\frac{i}{2B}\left(2Ak(b-k)\right)}\left[SAST^{{\mathbf N}}_\Psi({\mathcal M}_{\frac{Ak}{B}} f)\right](a,b-k).
\end{eqnarray*}
(iv)From (\ref{eqn qpst2}), we have for any positive real number $\sigma,$
\begin{eqnarray*}
&&\left[SAST^{{\mathbf N}}_\Psi(f(\sigma t))\right](a,b)\\
&&=\frac{a}{\sqrt{2\pi}}K_B\int_{\mathbb R} f(\sigma t)\overline{\Psi(a(t-b))}e^{-iat-\frac{i}{2B}\left(2At(b-t)-Dp^2\right)}dt\\
&&=\frac{a}{\sigma\sqrt{2\pi}}K_B\int_{\mathbb R} f(y)\overline{\Psi\left(\frac{a(y-\sigma b)}{\sigma}\right)}e^{-ia\frac{y}{\sigma}-\frac{i}{2B}\left(2A\frac{y}{\sigma}\left(b-\frac{y}{\sigma}\right)-Dp^2\right)}dy\\
&&=\frac{\frac{a}{\sigma}}{\sqrt{2\pi}}K_B\int_{\mathbb R} f(y)\overline{\Psi\left(\frac{a}{\sigma}(y-\sigma b)\right)}e^{-i\frac{a}{\sigma}{y}-\frac{i}{2\sigma^2B}\left(2{A}y\left(\sigma b-{y}\right)-\sigma^2Dp^2\right)}dy\\
&&=\left[SAST^{\bar{{\mathbf N}}}_\Psi f\right]\left(\frac{a}{\sigma}, \sigma b\right),\quad where \quad \bar{{\mathbf N}}=(A,\sigma^2 B,C,\sigma^2D;p,q)\\
\end{eqnarray*}
(iv)We have from (\ref{eqn qpst2})
\begin{eqnarray*}
&&\left[SAST^{{\mathbf N}}_\Psi(f(- t))\right](a,b)\\
&&=\frac{a}{\sqrt{2\pi}}\int_{\mathbb R} f(- t)\overline{\Psi(a(t-b))}e^{-iat-\frac{i}{2B}\left(2At(b-t)-Dp^2\right)}dt\\
&&=\frac{-a}{\sqrt{2\pi}}\int_{\mathbb R} f(y)\overline{\Psi(a(-y-b))}e^{-ia(-y)-\frac{i}{2B}\left(2A(-y)(b-(-y))-Dp^2\right)}dt\\
&&=\frac{-a}{\sqrt{2\pi}}\int_{\mathbb R} f(y)\overline{\Psi[(-a)(y-(-b))]}e^{-i(-a)y-\frac{i}{2B}\left(2Ay((-b)-y)-Dp^2\right)}dt\\
&&=\left[SAST^{{\mathbf N}}_\Psi f\right](-a,-b).\\
\end{eqnarray*}
This completes the proof.\\
\end{proof}

Next,  we  derive the admissibility condition for a window function $\Psi\in L^2(\mathbb R)$ via SAFT. Then
 we shall study some important theorems including the Rayleigh’s energy theorem, inversion
formula and   range theorem pertaining to the special affine Stockwell transform (\ref{eqn qpst}). To facilitate the motive, in the
preliminary step, we shall present the following lemma.

\begin{lemma}\label{lem ucp}Let $[SAST^{{\mathbf N}}_\Psi f](a,b)$  and $\mathbf S_{{\mathbf N}}[f(t)]$ be the special affine Stockwell transform  and special affine Fourier transform of a signal $f(t)$ in $L^2(\mathbb R)$. Then, we have
\begin{eqnarray}
 \nonumber&&\big|\mathbf S_{{\mathbf N}}\big[SAST^{\mathbf N}_\Psi f](a,b)\big](w)\big|\\
 \label{eqn lem ucp}&&=\frac{1}{\sqrt{2\pi}}\left|\mathbf S_{{\mathbf N}}[f]\big(w+{a}{B}\big)\right|.\left|\overline{\mathbf S_{{\mathbf N}} \left[ e^{i\left\{-y\left(1-\frac{1}{a}\right)\frac{p}{B}-\frac{Ay^2}{2B}\left(1+\frac{1}{a^2}\right)\right\}}\Psi(y)\right]}\left(\frac{w}{a}\right)\right|.
\end{eqnarray}
where $\Psi$ is window function in $L^2(\mathbb R)$.
\end{lemma}
 \begin{proof}
Invoking the convolution theorem of SAFT (\ref{eqn qpft conv}), we can compute the SAFT of $[SAST^\mu_\Psi f](a,b)$, as

\begin{eqnarray}
\nonumber \left|\mathbf S_{{\mathbf N}}\left[[SAST^{\mathbf N}_\Psi f](a,b)\right](w)\right|&=&\left|\frac{1}{\sqrt{2\pi}}\mathbf S_{{\mathbf N}}\left[(\mathcal M_{-a}f)\Theta_{{\mathbf N}} (\mathcal D_{a}\tilde\Psi)\right](w)\right|\\
\nonumber&=&\left|\frac{1}{\sqrt{2\pi}}e^{\frac{i}{2 B}\left(2w(Dp-Bq)-dw^2\right)}\mathbf S_{{\mathbf N}}\left[(\mathcal M_{-a}f)\right](w)\mathbf S_{{\mathbf N}}\left[(\mathcal D_{a}\tilde\Psi)\right](w)\right|\\
\label{eqn ucp lem1}&=&\frac{1}{\sqrt{2\pi}}\big|\mathbf S_{{\mathbf N}}\big[(\mathcal M_{-a}f)\big(w)\big|.\big|\mathbf S_{{\mathbf N}}\left[(\mathcal D_{a}\tilde\Psi)\right](w)\big|.
\end{eqnarray}
Now,

\begin{eqnarray}
\nonumber&&\mathbf S_{{\mathbf N}}\left[(\mathcal M_{-a}f)\right](w)\\
\nonumber&&\qquad=K_B\int_{\mathbb R} e^{-iat} f(t) e^{\frac{i}{2B}\left(At^2+2t(p-w)-2w(Dp-Bq)+D(w^2+p^2)\right)}dt\\
\nonumber&&\qquad=K_B\int_{\mathbb R}  f(t) e^{\frac{i}{2B}\left(At^2+2t(p-(w+aB))-2w(Dp-Bq)+D(w^2+p^2)\right)}dt\\
\nonumber&&\qquad=e^{\frac{i}{2B}\left(2aB(Dp-Bq)-D(a^2B^2+2waB)\right)}\\
\nonumber&&\qquad\qquad\qquad\times K_B\int_{\mathbb R}  f(t) e^{\frac{i}{2B}\left(At^2+2t(p-(w+aB))-2(w+aB)(Dp-Bq)+D((w+aB)^2+p^2)\right)}dt\\
\label{upc lem2}&&\qquad=e^{\frac{i}{2}\left(2a(Dp-Bq)-Da(aB+2w)\right)}\mathbf S_{{\mathbf N}}[f]\left(w+{a}{B}\right)
\end{eqnarray}

And
\begin{eqnarray}
\nonumber&&\mathbf S_{{\mathbf N}}\left[(\mathcal D_{a}\tilde\Psi)\right](w)\\
\nonumber&&\qquad=K_B\int_{\mathbb R} |a|\overline{\Psi(-at)} e^{\frac{i}{2B}\left(At^2+2t(p-w)-2w(Dp-Bq)+D(w^2+p^2)\right)}dt\\
\nonumber&&\qquad=K_B\int_{\mathbb R} \overline{\Psi(y)} e^{\frac{i}{2B}\left\{A\left(\frac{-y}{a}\right)^2+2\left(\frac{-y}{a}\right)(p-w)-2w(Dp-Bq)+D(w^2+p^2)\right\}}dy\\
\nonumber&&\qquad=K_B\int_{\mathbb R} \overline{\Psi(y)} e^{\frac{-i}{2B}\left\{Ay^2+2y\left(p-\frac{w}{a}\right)-2\left(\frac{w}{a}\right)(Dp-Bq)+D\left(\left(\frac{w}{a}\right)^2+p^2\right)\right\}}\\
\nonumber&&\qquad\qquad\qquad\times e^{\frac{i}{2B}\left\{A\left(\frac{y}{a}\right)^2+Ay^2-2\left(\frac{y}{a}\right)(p-w)+2y\left(p-\frac{w}{a}\right)-2w(Dp-Bq)-2\frac{w}{a}(Dp-Bq)+D(w^2+p^2)+D\left(\left(\frac{w}{a}\right)^2+p^2\right)\right\}}dy\\
\nonumber&&\qquad=e^{\frac{i}{2B}\left\{Dw^2\left(1+\frac{1}{a^2}\right)+2Dp^2-2w(Dp-Bq)\left(1+\frac{1}{a}\right)\right\}}K_B \int_{\mathbb R} \overline{\Psi(y)}e^{\frac{-i}{2B}\left\{-Ay^2\left(1+\frac{1}{a^2}\right)-2yp\left(1-\frac{1}{a}\right)\right\}} \\
\nonumber&&\qquad\qquad\qquad\times e^{\frac{-i}{2B}\left\{Ay^2+2y\left(p-\frac{w}{a}\right)-2\left(\frac{w}{a}\right)(Dp-Bq)+D\left(\left(\frac{w}{a}\right)^2+p^2\right)\right\}} dy\\
\nonumber&&\qquad=\frac{1}{K^*_B}e^{\frac{i}{2B}\left\{Dw^2\left(1+\frac{1}{a^2}\right)+2Dp^2-2w(Dp-Bq)\left(1+\frac{1}{a}\right)\right\}}K_B\\
\nonumber&&\qquad\times K^*_B\overline{\int_{\mathbb R} \left[\Psi(y)e^{\frac{i}{2B}\left\{-Ay^2\left(1+\frac{1}{a^2}\right)-2yp\left(1-\frac{1}{a}\right)\right\}}\right] e^{\frac{i}{2B}\left\{Ay^2+2y\left(p-\frac{w}{a}\right)-2\left(\frac{w}{a}\right)(Dp-Bq)+D\left(\left(\frac{w}{a}\right)^2+p^2\right)\right\}} dz}\\
 \label{upc lem3}&&\qquad= \frac{K_B}{K^*_B}e^{\frac{i}{2B}\left\{Dw^2\left(1+\frac{1}{a^2}\right)+2Dp^2-2w(Dp-Bq)\left(1+\frac{1}{a}\right)\right\}}\overline{\mathbf S_{{\mathbf N}} \left[ e^{i\left\{-y\left(1-\frac{1}{a}\right)\frac{p}{B}-\frac{Ay^2}{2B}\left(1+\frac{1}{a^2}\right)\right\}}\Psi(y)\right]}\left(\frac{w}{a}\right)
\end{eqnarray}
Inserting (\ref{upc lem2}) and (\ref{upc lem3}) in (\ref{eqn ucp lem1}), we get the desired expression as
\begin{eqnarray*}
 \nonumber&&\left|\mathbf S_{{\mathbf N}}\left[[SAST^{\mathbf N}_\Psi f](a,b)\right](w)\right|\\
\nonumber&&=\frac{1}{\sqrt{2\pi}}\left|\mathbf S_{{\mathbf N}}[f]\left(w+{a}{B}\right)\right|.\left|\overline{\mathbf S_{{\mathbf N}} \left[ e^{i\left\{-y\left(1-\frac{1}{a}\right)\frac{p}{B}-\frac{Ay^2}{2B}\left(1+\frac{1}{a^2}\right)\right\}}\Psi(y)\right]}\left(\frac{w}{a}\right)\right|.
\end{eqnarray*}
Which completes proof.\\
\end{proof}

We are now in a position to present admissibility condition for proposed SAST.
\begin{theorem}[Admissibility condition]\label{thm adm}
A window function $\Psi\in L^2(\mathbb R)$ is admissible in the special affine Stockwell domain if the following condition
holds:
\begin{equation}\label{3.16}
\mathcal C_\Psi=\int_{\mathbb R^+}\left|{\mathbf S_{{\mathbf N}} \left[ e^{i\left\{t\left(1-\left(1-\frac{1}{a}\right)\frac{p}{B}\right)-\frac{At^2}{2B}\left(1+\frac{1}{a^2}\right)\right\}}\Psi(t)\right]}\left(\frac{w}{a}\right)\right|^2da<\infty \quad a.e \quad w\in \mathbb R,
\end{equation}

where $\mathcal C_\Psi$ is real and positive.

\end{theorem}

\begin{proof}
By the  virtue of Lemma \ref{lem ucp}, translation invariance of the Lebesgue integral and modulation property of SAFT,
 for any square integrable function $f$, we have
\begin{eqnarray}\label{ad1}
\nonumber&&\int_{\mathbb R}\int_{\mathbb R^+}\left|\left\langle f,\,\Psi_{{{\mathbf N}}, a,b}\right\rangle\right|^2dadb\\
\nonumber&&\qquad=\int_{\mathbb R}\int_{\mathbb R^+}\left|\left(\left[\left(\mathcal M_{-a}f\right)\Theta_{{\mathbf N}} \left(\mathcal D_{a}\tilde\Psi\right)\right](b)\right)(b)\right|^2dadb\\
\nonumber&&\qquad=\int_{\mathbb R}\int_{\mathbb R^+}\left|\mathbf S_{{\mathbf N}}\left[\left(\mathcal M_{-a}f\right)\Theta_{{\mathbf N}} \left(\mathcal D_{a}\tilde\Psi\right)\right](w)\right|^2dadw\\
\nonumber&&\qquad=\int_{\mathbb R}\int_{\mathbb R^+}\left|\mathbf S_{{\mathbf N}}\left[\left(\mathcal M_{-a}f\right)\right](w)\mathbf S_{{\mathbf N}}\left[ \left(\mathcal D_{a}\tilde\Psi\right)\right](w)\right|^2dadw\\
\nonumber&&\qquad=\int_{\mathbb R}\int_{\mathbb R^+}\left|\mathbf S_{{\mathbf N}}[f]\left(w+{a}{B}\right)\overline{\mathbf S_{{\mathbf N}} \left[ e^{i\left\{-y\left(1-\frac{1}{a}\right)\frac{p}{B}-\frac{Ay^2}{2B}\left(1+\frac{1}{a^2}\right)\right\}}\Psi(y)\right]}\left(\frac{w}{a}\right)\right|^2dadw\\
\nonumber&&\qquad=\int_{\mathbb R}\left|\mathbf S_{{\mathbf N}}[f]\left(w\right)\right|^2\left(\int_{\mathbb R^+}\left|\mathbf S_{{\mathbf N}} \left[ e^{i\left\{-y\left(1-\frac{1}{a}\right)\frac{p}{B}-\frac{Ay^2}{2B}\left(1+\frac{1}{a^2}\right)\right\}}\Psi(y)\right]\left(\frac{w-a}{a}\right)\right|^2da\right)dw\\
\nonumber&&\qquad=\int_{\mathbb R}\left|\mathbf S_{{\mathbf N}}[f]\left(w\right)\right|^2\left(\int_{\mathbb R^+}\left|{\mathbf S_{{\mathbf N}} \left[ e^{i\left\{-y\left(1-\frac{1}{a}\right)\frac{p}{B}-\frac{Ay^2}{2B}\left(1+\frac{1}{a^2}\right)\right\}}e^{iy}\Psi(y)\right]}\left(\frac{w}{a}\right)\right|^2da\right)dw\\
\nonumber&&\qquad=\int_{\mathbb R}\left|\mathbf S_{{\mathbf N}}[f]\left(w\right)\right|^2\left(\int_{\mathbb R^+}\left|{\mathbf S_{{\mathbf N}} \left[ e^{i\left\{y\left(1-\left(1-\frac{1}{a}\right)\frac{p}{B}\right)-\frac{Ay^2}{2B}\left(1+\frac{1}{a^2}\right)\right\}}\Psi(y)\right]}\left(\frac{w}{a}\right)\right|^2da\right)dw.\\
\end{eqnarray}

For $f=\Psi,$ (\ref{ad1}) yields
\begin{eqnarray}\label{ad2}
\nonumber&&\int_{\mathbb R}\int_{\mathbb R^+}\left|\left\langle \Psi,\,\Psi_{{{\mathbf N}}, a,b}\right\rangle\right|^2dadb\\
\nonumber&&\qquad=\int_{\mathbb R}\left|\mathbf S_{{\mathbf N}}[\Psi]\left(w\right)\right|^2\left(\int_{\mathbb R^+}\left|{\mathbf S_{{\mathbf N}} \left[ e^{i\left\{y\left(1-\left(1-\frac{1}{a}\right)\frac{p}{B}\right)-\frac{Ay^2}{2B}\left(1+\frac{1}{a^2}\right)\right\}}\Psi(y)\right]}\left(\frac{w}{a}\right)\right|^2da\right)dw.\\
\end{eqnarray}

Since $\Psi\in L^2(\mathbb R)$, thus, we conclude that (\ref{ad2}) is non-zero and finite, provided
\begin{equation*}
0<\mathcal C_\Psi=\int_{\mathbb R^+}\left|{\mathbf S_{{\mathbf N}} \left[ e^{i\left\{t\left(1-\left(1-\frac{1}{a}\right)\frac{p}{B}\right)-\frac{At^2}{2B}\left(1+\frac{1}{a^2}\right)\right\}}\Psi(t)\right]}\left(\frac{w}{a}\right)\right|^2da<\infty \quad a.e \quad w\in \mathbb R,
\end{equation*}
which is the desired admissibility condition.
\end{proof}

\begin{theorem}[Rayleigh’s Energy Theorem]\label{thm moyal}If $\Psi \in L^2(\mathbb R)$ is admissible window and $[SAST^{{\mathbf N}}_{\Psi}f(t)](a,b)$ is the special affine Stockwell transform of any signal $f$ defined in (\ref{eqn qpst}),then for any $f,g \in L^2(\mathbb R)$, we have
\begin{equation}\label{eqn moyl}
\int_{\mathbb R}\int_{\mathbb R^+}[SAST^{{\mathbf N}}_{\Psi}f(t)](a,b)\overline{[SAST^{{\mathbf N}}_{\Psi}g(t)](a,b)}dadb=\frac{\mathcal C_\Psi}{2\pi}\langle f,g\rangle,
\end{equation}
where $\mathcal C_\Psi$ is given by (\ref{3.16}).
\end{theorem}

\begin{proof}
From Proposition \ref{prop 1}, the SAST of any  $f \in L^2(\mathbb R)$ is given by
\begin{eqnarray*}
\nonumber&&[SAST^{{\mathbf N}}_\Psi f(t)](a,b)\\
\nonumber&&\qquad=\frac{{e^{-iab}}K_B}{\sqrt{2\pi}K^*_B}\int_{\mathbb R}e^{\frac{-i}{2B}\left(2\frac{w}{a}(Dp-Bq)-D\left(\frac{w^2}{a^2}-2p^2 \right)\right)} \mathbf S_{{\mathbf N}}\left[f(t)\right](w)\\
&&\qquad\qquad\qquad\times\overline{\mathbf S_{{\mathbf N}}\left[ e^{i\left\{t\left(1-\left(1-\frac{1}{a}\right)\frac{p}{B}\right)-\frac{At^2}{2B}\left(1+\frac{1}{a^2}\right)\right\}}\Psi(t)\right]}\left(\frac{w}{a}\right)\overline{\mathcal K_{{\mathbf N}}(b,w)}dw.
\end{eqnarray*}
Similarly,
\begin{eqnarray*}
\nonumber&&[SAST^{{\mathbf N}}_\Psi g(t)](a,b)\\
\nonumber&&\qquad=\frac{{e^{-iab}}K_B}{\sqrt{2\pi}K^*_B}\int_{\mathbb R}e^{\frac{-i}{2B}\left(2\frac{\eta}{a}(Dp-Bq)-D\left(\frac{\eta^2}{a^2}-2p^2 \right)\right)} \mathbf S_{{\mathbf N}}\left[f(t)\right](\eta)\\
&&\qquad\qquad\qquad\times\overline{\mathbf S_{{\mathbf N}}\left[ e^{i\left\{t\left(1-\left(1-\frac{1}{a}\right)\frac{p}{B}\right)-\frac{At^2}{2B}\left(1+\frac{1}{a^2}\right)\right\}}\Psi(t)\right]}\left(\frac{\eta}{a}\right)\overline{\mathcal K_{{\mathbf N}}(b,\eta)}d\eta.
\end{eqnarray*}
Then, an implication of the Fubini’s theorem yields
\begin{eqnarray*}
&&\int_{\mathbb R}\int_{\mathbb R^+}[SAST^{{\mathbf N}}_{\Psi}f(t)](a,b)\overline{[SAST^{{\mathbf N}}_{\Psi}g(t)](a,b)}dadb\\
&&=\frac{1}{{2\pi}}\int_{\mathbb R}\int_{\mathbb R^+}\left\{\int_{\mathbb R}e^{\frac{-i}{2B}\left(2\frac{w}{a}(Dp-Bq)-D\left(\frac{w^2}{a^2}-2p^2 \right)\right)} \mathbf S_{{\mathbf N}}\left[f(t)\right](w)\right.\\
&&\qquad\qquad\qquad\times\left.\overline{\mathbf S_{{\mathbf N}}\left[ e^{i\left\{t\left(1-\left(1-\frac{1}{a}\right)\frac{p}{B}\right)-\frac{At^2}{2B}\left(1+\frac{1}{a^2}\right)\right\}}\Psi(t)\right]}\left(\frac{w}{a}\right)\overline{\mathcal K_{{\mathbf N}}(b,w)}dw\right\}\\
&&\qquad\qquad\times \left\{\int_{\mathbb R}e^{\frac{i}{2B}\left(2\frac{\eta}{a}(Dp-Bq)-D\left(\frac{\eta^2}{a^2}-2p^2 \right)\right)}\overline{ \mathbf S_{{\mathbf N}}\left[f(t)\right](\eta)}\right.\\
&&\qquad\qquad\qquad\left.\times{\mathbf S_{{\mathbf N}}\left[ e^{i\left\{t\left(1-\left(1-\frac{1}{a}\right)\frac{p}{B}\right)-\frac{At^2}{2B}\left(1+\frac{1}{a^2}\right)\right\}}\Psi(t)\right]}\left(\frac{\eta}{a}\right){\mathcal K_{{\mathbf N}}(b,\eta)}d\eta \right\}dadb\\
&&=\frac{1}{{2\pi}}\int_{\mathbb R}\int_{\mathbb R}\int_{\mathbb R^+}e^{\frac{i}{2B}\left\{\frac{2}{a}(Dp-Bq)\left(\eta-w\right)-2\frac{D}{a^2}(\eta^2-w^2)\right\}} \mathbf S_{{\mathbf N}}\left[f(t)\right](w)\overline{\mathbf S_{{\mathbf N}}\left[g(t)\right]}(\eta)\\
&&\times \overline{\mathbf S_{{\mathbf N}}\left[ e^{i\left\{t\left(1-\left(1-\frac{1}{a}\right)\frac{p}{B}\right)-\frac{At^2}{2B}\left(1+\frac{1}{a^2}\right)\right\}}\Psi(t)\right]}\left(\frac{w}{a}\right){\mathbf S_{{\mathbf N}}\left[ e^{i\left\{t\left(1-\left(1-\frac{1}{a}\right)\frac{p}{B}\right)-\frac{At^2}{2B}\left(1+\frac{1}{a^2}\right)\right\}}\Psi(t)\right]}\left(\frac{\eta}{a}\right)\\
&&\qquad\times\left\{{\mathcal K_{{\mathbf N}}(b,\eta)}\overline{\mathcal K_{{\mathbf N}}(b,w)}db \right\}dadw d\eta\\
&&=\frac{1}{2\pi}\int_{\mathbb R}\int_{\mathbb R}\int_{\mathbb R^+} \mathbf S_{{\mathbf N}}\left[f(t)\right](w)\overline{\mathbf S_{{\mathbf N}}\left[g(t)\right]}(\eta)\overline{\mathbf S_{{\mathbf N}}\left[ e^{i\left\{t\left(1-\left(1-\frac{1}{a}\right)\frac{p}{B}\right)-\frac{At^2}{2B}\left(1+\frac{1}{a^2}\right)\right\}}\Psi(t)\right]}\left(\frac{w}{a}\right)\\
&&\times {\mathbf S_{{\mathbf N}}\left[e^{i\left\{t\left(1-\left(1-\frac{1}{a}\right)\frac{p}{B}\right)-\frac{At^2}{2B}\left(1+\frac{1}{a^2}\right)\right\}}\Psi(t)\right]}\left(\frac{\eta}{a}\right)e^{\frac{i}{2B}\left\{\frac{2}{a}(Dp-Bq)\left(\eta-w\right)-2\frac{D}{a^2}(\eta^2-w^2)\right\}} {\delta}(\eta-w)dadw d\eta\\
&&=\frac{1}{2\pi}\int_{\mathbb R}\int_{\mathbb R^+} \mathbf S_{{\mathbf N}}\left[f(t)\right](w)\overline{\mathbf S_{{\mathbf N}}\left[g(t)\right]}(w)
\overline{\mathbf S_{{\mathbf N}}\left[e^{i\left\{t\left(1-\left(1-\frac{1}{a}\right)\frac{p}{B}\right)-\frac{At^2}{2B}\left(1+\frac{1}{a^2}\right)\right\}}\Psi(t)\right]}\left(\frac{w}{a}\right)\\
&&\times {\mathbf S_{{\mathbf N}}\left[e^{i\left\{t\left(1-\left(1-\frac{1}{a}\right)\frac{p}{B}\right)-\frac{At^2}{2B}\left(1+\frac{1}{a^2}\right)\right\}}\Psi(t)\right]}\left(\frac{w}{a}\right)dadw dw\\
&&=\frac{1}{2\pi}\int_{\mathbb R}\int_{\mathbb R^+} \mathbf S_{{\mathbf N}}\left[f(t)\right](w)\overline{\mathbf S_{{\mathbf N}}\left[g(t)\right]}(w)
\left|{\mathbf S_{{\mathbf N}}\left[e^{i\left\{t\left(1-\left(1-\frac{1}{a}\right)\frac{p}{B}\right)-\frac{At^2}{2B}\left(1+\frac{1}{a^2}\right)\right\}}\Psi(t)\right]}\left(\frac{w}{a}\right)\right|^2da dw\\
&&=\frac{1}{2\pi}\int_{\mathbb R} \mathbf S_{{\mathbf N}}\left[f(t)\right](w)\overline{\mathbf S_{{\mathbf N}}\left[g(t)\right]}(w)\left\{\int_{\mathbb R^+}
\left|{\mathbf S_{{\mathbf N}}\left[e^{i\left\{t\left(1-\left(1-\frac{1}{a}\right)\frac{p}{B}\right)-\frac{At^2}{2B}\left(1+\frac{1}{a^2}\right)\right\}}\Psi(t)\right]}\left(\frac{w}{a}\right)\right|^2 da\right\}dw\\
&&=\frac{1}{2\pi}\times \mathcal C_\Psi \int_{\mathbb R} \mathbf S_{{\mathbf N}}\left[f(t)\right](w)\overline{\mathbf S_{{\mathbf N}}\left[g(t)\right]}(w)dw\\
&&= \frac{\mathcal C_\Psi}{2\pi}\left\langle \mathbf S_{{\mathbf N}}\left[f(t)\right](w),{\mathbf S_{{\mathbf N}}\left[g(t)\right]}(w)\right\rangle\\
&&=\frac{\mathcal C_\Psi}{2\pi}\langle f,g\rangle.
\end{eqnarray*}
This completes the proof

\end{proof}

\begin{corollary}\label{cor1}For $f=g$, the above theorem yields energy preserving relation
\begin{equation}
\left\|[SAST^{{\mathbf N}}_{\Psi}f(t)](a,b)\right\|^2=\frac{\mathcal C_\Psi}{2\pi}\|f\|^2.
\end{equation}
\end{corollary}

Next, we shall obtain an inversion formula associated with the special affine Stockwell transform (\ref{eqn qpst}).

\begin{theorem}[Inversion Formula).]\label{thm inver}Every signal $f \in L^2(\mathbb R)$ can be reconstructed from
the corresponding  special affine Stockwell transform $[SAST^{{\mathbf N}}_{\Psi}f(t)](a,b)$ by the formula
\begin{equation}\label{eqn thm inv}
f(t)=\frac{\sqrt{2\pi}}{\mathcal C_\Psi}\int_{\mathbb R}\int_{\mathbb R^+}[SAST^{{\mathbf N}}_{\Psi}f(t)](a,b)\Psi_{{{\mathbf N}},a,b}(t)dadb, \quad a.e.
\end{equation}
\end{theorem}

\begin{proof}
 For any arbitrary $g\in L^2(\mathbb R)$, Rayleigh’s identity (\ref{eqn moyl}) yields
\begin{eqnarray*}
\frac{\mathcal C_\Psi}{2\pi}\langle f,g\rangle&=&\int_{\mathbb R}\int_{\mathbb R^+}[SAST^{{\mathbf N}}_{\Psi}f(t)](a,b)\overline{[SAST^{{\mathbf N}}_{\Psi}g(t)](a,b)}dadb\\
&=&\frac{1}{\sqrt{2\pi}}\int_{\mathbb R}\int_{\mathbb R^+}[SAST^{{\mathbf N}}_{\Psi}f(t)](a,b)\left\{\overline{\int_{\mathbb R}g(t)\overline{\Psi_{{{\mathbf N}},a,b}(t)}dt}\right\}dadb\\
&=&\frac{1}{\sqrt{2\pi}}\int_{\mathbb R}\int_{\mathbb R}\int_{\mathbb R^+}[SAST^{{\mathbf N}}_{\Psi}f(t)](a,b)\overline{g(t)}\Psi_{{{\mathbf N}},a,b}(t)dadbdt\\
&=&\frac{1}{\sqrt{2\pi}}\left\langle\int_{\mathbb R}\int_{\mathbb R^+}[SAST^{{\mathbf N}}_{\Psi}f(t)](a,b)\Psi_{{{\mathbf N}},a,b}(t)dadb, \quad{g(t)}\right\rangle.\\
\end{eqnarray*}
As $g\in L^2(\mathbb R)$ is arbitrary, it follows that
\begin{equation*}
f(t)=\frac{\sqrt{2\pi}}{\mathcal C_\Psi}\int_{\mathbb R}\int_{\mathbb R^+}[SAST^{{\mathbf N}}_{\Psi}f(t)](a,b)\Psi_{{{\mathbf N}},a,b}(t)dadb, \quad a.e.
\end{equation*}
This completes the proof.
\end{proof}

Next, we shall present a complete characterization of the range of the proposed special affine
Stockwell transform (\ref{eqn qpst}). The structure of the range of the SAST is obtained by using the
inversion formula (\ref{eqn thm inv}) and the well known Fubini's theorem.

\begin{theorem}[Characterization of range of $SAST^{{\mathbf N}}_\Psi$ ] Let $\Psi$ be an admissible window function in $L^2(\mathbb R)$, then a function $f\in L^2(\mathbb R^+\times \mathbb R)$ lies in the range of $SAST^{{\mathbf N}}_\Psi(L^2(\mathbb R))$ if and only if
\begin{equation}\label{range eqn}
f(c,d)={\mathcal C^{-1}_\Psi}\int_{\mathbb R}\int_{\mathbb R^+}f(a,b)\langle\Psi_{{{\mathbf N}},a,b},\, \Psi_{{{\mathbf N}},c,d}\rangle da db.
\end{equation}
\end{theorem}
\begin{proof}
Let $f\in SAST^{{\mathbf N}}_\Psi(L^2(\mathbb R))$, then there is function $g\in L^2(\mathbb R)$ such that\\
 $[SAST^{{\mathbf N}}_\Psi g]=f.$ We claim that $f$ satisfies (\ref{range eqn}). To show this, we proceed as
 \begin{eqnarray*}
 f(c,d)&=&[SAST^{{\mathbf N}}_\Psi g](c,d)\\
 &=&\frac{1}{\sqrt{2\pi}}\int_{\mathbb R}g(t)\overline{\Psi_{{{\mathbf N}},c,d}(t)}dt\\
  &=&\frac{1}{\sqrt{2\pi}}\int_{\mathbb R}\left\{ \frac{\sqrt{2\pi}}{\mathcal C_\Psi}\int_{\mathbb R}\int_{\mathbb R^+}[SAST^{{\mathbf N}}_{\Psi}g(t)](a,b)\Psi_{{{\mathbf N}},a,b}(t)dadb \right\}\overline{\Psi_{{{\mathbf N}},c,d}(t)}dt\\
   &=&{\mathcal C^{-1}_\Psi}\int_{\mathbb R}\int_{\mathbb R^+}[SAST^{{\mathbf N}}_{\Psi}g(t)](a,b)\left\{ \int_{\mathbb R}\Psi_{{{\mathbf N}},a,b}(t)\overline{\Psi_{{{\mathbf N}},c,d}(t)}dt \right\}dadb\\
   &=&{\mathcal C^{-1}_\Psi}\int_{\mathbb R}\int_{\mathbb R^+}f(a,b)\langle\Psi_{{{\mathbf N}},a,b},\, \Psi_{{{\mathbf N}},c,d}\rangle da db.
 \end{eqnarray*}
 Conversely, let $f\in L^2(\mathbb R^+\times \mathbb R)$ satisfies (\ref{range eqn}). Then, we shall prove that $f\in SAST^{{\mathbf N}}_\Psi(L^2(\mathbb R))$; i.e., there must exists $g\in L^2(\mathbb R)$ such that $[SAST^{{\mathbf N}}_\Psi g]=f.$ We claim that
 \begin{equation*}
 g(t)= \frac{\sqrt{2\pi}}{\mathcal C_\Psi}\int_{\mathbb R}\int_{\mathbb R^+}f(a,b)\Psi_{{{\mathbf N}},a,b}(t)dadb.
 \end{equation*}
 It is easy to prove that $\|g\|^2=\mathcal C^{-1}_\Psi\|f\|^2<\infty,$ implies $g\in L^2(\mathbb R).$\\
 Moreover, by the virtue of Fubini’s theorem, we have
 \begin{eqnarray*}
 [SAST^{{\mathbf N}}_\Psi g](c,d)&=&\frac{1}{\sqrt{2\pi}}\int_{\mathbb R}g(t)\overline{\Psi_{{{\mathbf N}},c,d}(t)}dt\\
 &=&\frac{1}{\sqrt{2\pi}}\int_{\mathbb R}\left\{ \frac{\sqrt{2\pi}}{\mathcal C_\Psi}\int_{\mathbb R}\int_{\mathbb R^+}f(a,b)\Psi_{{{\mathbf N}},a,b}(t)dadb\right\}\overline{\Psi_{{{\mathbf N}},c,d}(t)}dt\\
 &=&{\mathcal C^{-1}_\Psi}\int_{\mathbb R}\int_{\mathbb R^+}f(a,b)\langle\Psi_{{{\mathbf N}},a,b},\,\Psi_{{{\mathbf N}},c,d}\rangle dadb\\
 &=& f(c,d).
 \end{eqnarray*}
 This completes the proof.
\end{proof}
\section{The relationship between the Special Affine Scaled Wigner distribution and the SAST}\label{secrel}
The Wigner distribution (WD), which has gained much popularity in recent years  as it has emerged as the premier tool for the analysis of both mono-component and bi-component non-stationary LFM signals \cite{st52,st53}. It is well known that
 the WD offers  perfect localization (localized on a straight line) to the mono-component LFM signals but while dealing with  multi-component signals cross terms appear because of their quadratic nature.
  Zhang et al.\cite{zhswd} introduced a scaled variant of the conventional  WD known as the scaled Wigner distribution (SWD) which is parameterized by a constant $k\in \mathbb {Q^+}$. The SWD gives a
novel way for the improvement of  the cross-term reduction  time–frequency resolution and angle resolution when
dealing with the multi-component LFM signal.
Recently Bhat and Dar \cite{scale} introduced the special affine scaled Wigner distribution (SASWD) by extending the
WD associated with special affine Fourier transform to the
novel one. The SASWD can be productive for signal processing theory and applications especially for detection and estimation of LFM signals.Thus it is significant to understand the its relationships with the neoteric time-frequency tools. Taking this opportunity, we shall  establish a fundamental relationship between the proposed SAST  and SASWD.
\begin{definition}\cite{scale}
Let $ {{\mathbf N}}=(A,B,C,D;p,q)$ be the matrix parameter and $f(t)$ be the finite energy signal in $L^2(\mathbb R),$ then the special affine scaled Wigner distribution (SASWD) is defined by
\begin{equation}\label{wvd}
\mathcal W^{{{\mathbf N}},k}_{f}(t,u)=\frac{1}{2\pi B}\int_{\mathbb R}f\left(t+k\frac{{\mathbf \tau}}{2}\right)\overline{f\left(t-k\frac{\tau}{2}\right)}e^{i\frac{\tau}{B}(Akt+kp-u)}d\tau,\, k\in\mathbb Q^+\\
\end{equation}
where $f\left(t+k\frac{{\mathbf \tau}}{2}\right)\overline{f\left(t-k\frac{\tau}{2}\right)}$ is the fractional instantaneous auto-correlation
function.\\
Analogous to (\ref{wvd}), the  cross-special affine scaled Wigner distribution of
the finite energy signals $f$ and $g$ can be defined as
\begin{equation}\label{cross wvd}
\mathcal W^{{{\mathbf N}},k}_{f,g}(t,u)=\frac{1}{2\pi B}\int_{\mathbb R}f\left(t+k\frac{{\mathbf \tau}}{2}\right)\overline{g\left(t-k\frac{\tau}{2}\right)}e^{i\frac{\tau}{B}(Akt+kp-u)}d\tau,\\
\end{equation}
where $k\in\mathbb Q^+.$
\end{definition}

In the following theorem, we develop a relationship between the cross-special affine scaled Wigner distribution  (\ref{cross wvd}) and the
proposed special affine Stockwell transform (\ref{eqn def qpst}).

\begin{theorem}Let $\mathcal W^{{{\mathbf N}},k}_{f,g}$ be  the cross-special affine scaled Wigner distribution of two signals $f,g\in L^2(\mathbb R)$ and $[SAST^{{\mathbf N}}_\Psi f](a,b)$ be the SAST of $f$ with respect  an  admissible window function $\Psi\in L^2(\mathbb R)$. Then we have
\begin{eqnarray}\label{eqn th wvd}
\nonumber&&\mathcal W^{{{\mathbf N}},k}_{f,g}(t,u)\\
\nonumber&&=\frac{2 e^{-i\frac{2}{Bk}[3Akt+kp-u]t}}{Bk.\mathcal C_\Psi}\int_{\mathbb R}\int_{\mathbb R^+}[SAST^{{\mathbf N}}_\Psi({\mathcal M}_\frac{-2u}{Bk}F)](a,b)\\
\nonumber&&\qquad\qquad\qquad\qquad\qquad\qquad\times\overline{[SAST^{{\mathbf N}}_\Psi ({\mathcal M}_{2\frac{A}{B}t}g)](-a,-(b+2t))}e^{-2it(a+\frac{A}{B}b)}dadb,\\
\end{eqnarray}
where $F(y)=e^{i\frac{2}{B}(At+p)y}f(y).$
\end{theorem}

\begin{proof}
For any $f,g\in L^2(\mathbb R),$ the cross-special affine scaled Wigner distribution can be expressed as
\begin{eqnarray}
\nonumber\mathcal \mathcal W^{{{\mathbf N}},k}_{f,g}(t,u)=\frac{1}{2\pi B}\int_{\mathbb R}f\left(t+k\frac{{\mathbf \tau}}{2}\right)\overline{g\left(t-k\frac{\tau}{2}\right)}e^{i\frac{\tau}{B}(Akt+kp-u)}d\tau\\
\label{wvd 1}\buildrel\rm t+\frac{k\tau}{2}=y \over=\frac{e^{-i\frac{2}{Bk}(Akt+kp-u)t}}{\pi Bk}\int_{\mathbb R}f\left(y\right)\overline{g\left(2t-y\right)}e^{i\frac{2}{Bk}(Akt+kp-u)y}dy.
\end{eqnarray}
By utilizing the Theorem \ref{thm p1} and inversion formula \ref{eqn thm inv} of SAST, we get
\begin{eqnarray*}
&&g(2t-y)\\
&&=\frac{\sqrt{2\pi}}{\mathcal C_\Psi}\int_{\mathbb R}\int_{\mathbb R^+}[SAST^{{\mathbf N}}_\Psi g(2t-y)](a,b)\Psi_{{{\mathbf N}},a,b}(y)dadb\\
&&=\frac{\sqrt{2\pi}}{\mathcal C_\Psi}\int_{\mathbb R}\int_{\mathbb R^+}e^{2i[at+\frac{A}{B}t(b+2t)]}[SAST^{{\mathbf N}}_\Psi ({\mathcal M}_{-2\frac{A}{B}t}g(-y))](a,b+2t)\Psi_{{{\mathbf N}},a,b}(y)dadb\\
&&=\frac{\sqrt{2\pi}}{\mathcal C_\Psi}\int_{\mathbb R}\int_{\mathbb R^+}e^{2i[at+\frac{A}{B}t(b+2t)]}[SAST^{{\mathbf N}}_\Psi ({\mathcal M}_{2\frac{A}{B}t}g)](-a,-(b+2t))\Psi_{{{\mathbf N}},a,b}(y)dadb\\
\end{eqnarray*}
Substituting the above estimate in (\ref{wvd 1})and using Fubini’s theorem, it yields
\begin{eqnarray*}
&&\mathcal W^{{{\mathbf N}},k}_{f,g}(t,u)\\
&&=\frac{\sqrt{2\pi}e^{-i\frac{2}{Bk}(Akt+kp-u)t}}{\pi Bk.\mathcal C_\Psi}\int_{\mathbb R}e^{-4i\frac{A}{B}t^2}f\left(y\right)e^{i\frac{2}{Bk}(Akt+kp-u)y}\left\{\int_{\mathbb R}\int_{\mathbb R^+} e^{-2it(a+\frac{A}{B}b)}\right.\\
&&\qquad\qquad\qquad\times \left.\overline{[SAST^{{\mathbf N}}_\Psi ({\mathcal M}_{2\frac{A}{B}t}g)](-a,-(b+2t))}\overline{\Psi_{{{\mathbf N}},a,b}(y)}dadb\right\}dy\\
&&=\frac{\sqrt{2\pi}e^{-i\frac{2}{Bk}(Akt+kp-u)t}}{\pi Bk.\mathcal C_\Psi}\int_{\mathbb R}e^{-4i\frac{A}{B}t^2}\left\{e^{i\frac{2}{B}(At+p)y}f\left(y\right)\right\}e^{-i\frac{2u}{Bk}y}\left\{\int_{\mathbb R}\int_{\mathbb R^+} e^{-2it(a+\frac{A}{B}b)}\right.\\
&&\qquad\qquad\qquad\times \left.\overline{[SAST^{{\mathbf N}}_\Psi ({\mathcal M}_{2\frac{A}{B}t}g)](-a,-(b+2t))}\overline{\Psi_{{{\mathbf N}},a,b}(y)}dadb\right\}dy\\
&&=\frac{\sqrt{2\pi}e^{-i\frac{2}{Bk}[3Akt+kp-u]t}}{\pi Bk.\mathcal C_\Psi}\int_{\mathbb R}\int_{\mathbb R^+}e^{-2it(a+\frac{A}{B}b)}\overline{[SAST^{{\mathbf N}}_\Psi ({\mathcal M}_{2\frac{A}{B}t}g)](-a,-(b+2t))}\\
&&\qquad\qquad\qquad\times \sqrt{2\pi}\left\{\frac{1}{\sqrt{2\pi}}\int_{\mathbb R}\left({\mathcal M}_{\frac{-2u}{Bk}}\left\{e^{i\frac{2}{B}(At+p)y}f\right\}\right)(y) \overline{\Psi_{{{\mathbf N}},a,b}(y)}dy\right\}dadb\\
&&=\frac{2 e^{-i\frac{2}{Bk}[3Akt+kp-u]t}}{Bk.\mathcal C_\Psi}\int_{\mathbb R}\int_{\mathbb R^+}e^{-2it(a+\frac{A}{B}b)}\overline{[SAST^{{\mathbf N}}_\Psi ({\mathcal M}_{2\frac{A}{B}t}g)](-a,-(b+2t))}\\
&&\qquad\qquad\qquad\qquad\qquad\qquad\times [SAST^{{\mathbf N}}_\Psi({\mathcal M}_{\frac{-2u}{Bk}}F)](a,b)dadb\\
&&=\frac{2 e^{-i\frac{2}{Bk}[3Akt+kp-u]t}}{Bk.\mathcal C_\Psi}\int_{\mathbb R}\int_{\mathbb R^+}[SAST^{{\mathbf N}}_\Psi({\mathcal M}_\frac{-2u}{Bk}F)](a,b)\\
&&\qquad\qquad\qquad\qquad\qquad\qquad\qquad\times\overline{[SAST^{{\mathbf N}}_\Psi ({\mathcal M}_{2\frac{A}{B}t}g)](-a,-(b+2t))}e^{-2it(a+\frac{A}{B}b)}dadb,
\end{eqnarray*}
where $F(y)=e^{i\frac{2}{B}(At+p)y}f(y).$\\
This completes the proof.
\end{proof}

\section{Uncertainty principles for the quadratic-phase Stockwell transform }\label{secucp}
The SAST can be regarded as the generalized version of the SAFT which has wide applications in optics and signal processing. On the other hand the uncertainty principles (UCPs)viz: Heisenberg’s UCP, logarithmic UCP, Nazarov’s UCP and etc,  owns their specific form for every time-frequency tool so  they can not be avoided.
 Keeping in mind the
 the  contemporary developments in the theory of uncertainty principles, in this section we shall formulate certain novel
 UCPs associated with the proposed SAST $[SAST^{{\mathbf N}}_\Psi f](a,b),$
 when regarded as a function of variable $b$. These new UCPs are governed by the matrix parameter  ${{\mathbf N}}=(A,B,C,D;p,q)$   and
are fruitful to understand mutual relations among different transform domains. Let us start with Heisenberg's UCP  for the proposed SAST.

\begin{theorem}[Heisenberg’s UCP]\label{thm hsb ucp} Let $[SAST^{{\mathbf N}}_\Psi f](a,b)$ be the SAST of a signal $f(t)\in L^2(\mathbb R)$ with respect to the admissible window function $\Psi\in L^2(\mathbb R),$ then the following uncertainty inequality holds:
\begin{eqnarray}\label{eqn hsb thm}
\nonumber&&\left\{\int_{\mathbb R}\int_{\mathbb R^+}b^2\left|[SAST^{{\mathbf N}}_\Psi f](a,b)\right|^2dadb\right\}^{1/2}\left\{\int_{\mathbb R}w^2\left|\mathbf S_{{\mathbf N}}[f]\left(w\right)\right|^2dw\right\}^{1/2}\ge \frac{B.\sqrt{\mathcal C_\Psi}}{2}\|f\|^2.\\
\end{eqnarray}
\end{theorem}

\begin{proof}
The Heisenberg's UCP for the SAFT reads \cite{1102s}
\begin{equation}\label{1hsb qpft}
\left\{\int_{\mathbb R}t^2|f(t)|^2dt\right\}^{1/2}\left\{\int_{\mathbb R}w^2|\mathbf S_{{\mathbf N}}[f](w)|^2dw\right\}^{1/2}\ge \frac{B}{2}\int_{\mathbb R}|f(t)|^2dt
\end{equation}
Identifying $[SAST^{{\mathbf N}}_\Psi f](a,b)$ as a function of translation parameter $b$ and replacing $f(t)$ by $[SAST^{{\mathbf N}}_\Psi f](a,b)$ in (\ref{1hsb qpft}), we have
\begin{eqnarray}
\nonumber&&\left\{\int_{\mathbb R}b^2\left|[SAST^{{\mathbf N}}_\Psi f](a,b)\right|^2db\right\}^{1/2}\left\{\int_{\mathbb R}w^2|\mathbf S_{{\mathbf N}}\left[SAST^{{\mathbf N}}_\Psi f](a,b)\right](w)|^2dw\right\}^{1/2}\\
\label{2hsb qpft}&&\qquad\qquad\qquad\qquad\ge \frac{B}{2}\int_{\mathbb R}\left|[SAST^{{\mathbf N}}_\Psi f](a,b)\right|^2db.
\end{eqnarray}
First integrating (\ref{2hsb qpft}) with respect to $da$ and then using Corollary \ref{cor1} to RHS, we get
\begin{eqnarray}
\nonumber&&\int_{\mathbb R^+}\left[\left\{\int_{\mathbb R}b^2\left|[SAST^{{\mathbf N}}_\Psi f](a,b)\right|^2db\right\}^{1/2}\left\{\int_{\mathbb R}w^2|\mathbf S_{{\mathbf N}}\left[SAST^{{\mathbf N}}_\Psi f](a,b)\right](w)|^2dw\right\}^{1/2}\right]da\\
\nonumber&&\qquad\qquad\qquad\qquad\ge \frac{B}{2}\int_{\mathbb R^+}\int_{\mathbb R}\left|[SAST^{{\mathbf N}}_\Psi f](a,b)\right|^2dadb\\
\label{3hsb qpft}&&\qquad\qquad\qquad\qquad\ge \frac{B.\mathcal C_\Psi}{4\pi}\|f\|^2.
\end{eqnarray}
Implementing Cauchy-Schwartz's inequality and Fubini's theorem, (\ref{3hsb qpft}) yields
\begin{eqnarray}
\nonumber&&\left\{\int_{\mathbb R}\int_{\mathbb R^+}b^2\left|[SAST^{{\mathbf N}}_\Psi f](a,b)\right|^2dadb\right\}^{1/2}\left\{\int_{\mathbb R}\int_{\mathbb R^+}w^2|\mathbf S_{{\mathbf N}}\left[SAST^{{\mathbf N}}_\Psi f](a,b)\right](w)|^2dadw\right\}^{1/2}\\
\label{4hsb qpft}&&\qquad\qquad\qquad\qquad\ge \frac{B.\mathcal C_\Psi}{4\pi}\|f\|^2.
\end{eqnarray}
Using Lemma \ref{lem ucp} and invoking translation invariance of the Lebesgue integral and modulation property of SAFT, we have
\begin{eqnarray}
\nonumber&&\int_{\mathbb R}\int_{\mathbb R^+}w^2|\mathbf S_{{\mathbf N}}\left[SAST^{{\mathbf N}}_\Psi f](a,b)\right](w)|^2dadw\\
\nonumber&&\qquad=\frac{1}{\sqrt{2\pi}}\int_{\mathbb R}\int_{\mathbb R^+}w^2\left|\mathbf S_{{\mathbf N}}[f]\left(w+{a}{B}\right)\overline{\mathbf S_{{\mathbf N}} \left[ e^{i\left\{-t\left(1-\frac{1}{a}\right)\frac{p}{B}-\frac{At^2}{2B}\left(1+\frac{1}{a^2}\right)\right\}}\Psi(t)\right]}\left(\frac{w}{a}\right)\right|^2dadw\\
\nonumber&&\qquad=\frac{1}{\sqrt{2\pi}}\int_{\mathbb R}w^2\left|\mathbf S_{{\mathbf N}}[f]\left(w\right)\right|^2\left(\int_{\mathbb R^+}\left|{\mathbf S_{{\mathbf N}} \left[ e^{i\left\{-t\left(1-\frac{1}{a}\right)\frac{p}{B}-\frac{At^2}{2B}\left(1+\frac{1}{a^2}\right)\right\}}\Psi(t)\right]}\left(\frac{w-a}{a}\right)\right|^2da\right)dw\\
\nonumber&&\qquad=\frac{1}{\sqrt{2\pi}}\int_{\mathbb R}w^2\left|\mathbf S_{{\mathbf N}}[f]\left(w\right)\right|^2\left(\int_{\mathbb R^+}\left|{\mathbf S_{{\mathbf N}} \left[ e^{i\left\{-t\left(1-\frac{1}{a}\right)\frac{p}{B}-\frac{At^2}{2B}\left(1+\frac{1}{a^2}\right)\right\}}e^{it}\Psi(t)\right]}\left(\frac{w}{a}\right)\right|^2da\right)dw\\
\nonumber&&\qquad=\frac{1}{\sqrt{2\pi}}\int_{\mathbb R}w^2\left|\mathbf S_{{\mathbf N}}[f]\left(w\right)\right|^2\left(\int_{\mathbb R^+}\left|{\mathbf S_{{\mathbf N}} \left[ e^{i\left\{t\left(1-\left(1-\frac{1}{a}\right)\frac{p}{B}\right)-\frac{At^2}{2B}\left(1+\frac{1}{a^2}\right)\right\}}\Psi(t)\right]}\left(\frac{w}{a}\right)\right|^2da\right)dw\\
\label{5hsb qpft}&&\qquad\buildrel\rm (using (\ref{3.16})) \over=\frac{\mathcal C_\Psi}{\sqrt{2\pi}}\int_{\mathbb R}w^2\left|\mathbf S_{{\mathbf N}}[f]\left(w\right)\right|^2dw.
\end{eqnarray}
Inserting (\ref{5hsb qpft}) in (\ref{4hsb qpft}), we obtain

\begin{eqnarray*}
\left\{\int_{\mathbb R}\int_{\mathbb R^+}b^2\left|[SAST^{{\mathbf N}}_\Psi f](a,b)\right|^2dadb\right\}^{1/2}\left\{\frac{\mathcal C_\Psi}{\sqrt{2\pi}}\int_{\mathbb R}w^2\left|\mathbf S_{{\mathbf N}}[f]\left(w\right)\right|^2dw\right\}^{1/2}\ge \frac{B.\mathcal C_\Psi}{4\pi}\|f\|^2.
\end{eqnarray*}
On further simplification, above inequality yields the desired
result
\begin{eqnarray*}
\left\{\int_{\mathbb R}\int_{\mathbb R^+}b^2\left|[SAST^{{\mathbf N}}_\Psi f](a,b)\right|^2dadb\right\}^{1/2}\left\{\int_{\mathbb R}w^2\left|\mathbf S_{{\mathbf N}}[f]\left(w\right)\right|^2dw\right\}^{1/2}\ge \frac{B.\sqrt{\mathcal C_\Psi}}{2}\|f\|^2.
\end{eqnarray*}
This completes the proof.\\
\end{proof}

\begin{remark}From the UCP (\ref{eqn hsb thm}), we promptly infer that:\\
\begin{itemize}
\item For the parametric set ${{\mathbf N}}=(A,B,C,D;p,q),$ Theorem \ref{thm hsb ucp} yields Heisenberg UCP pertaining to linear canonical Stockwell transform :
\begin{eqnarray*}
\left\{\int_{\mathbb R}\int_{\mathbb R^+}b^2\left|[LCST^{{\mathbf N}}_\Psi f](a,b)\right|^2dadb\right\}\left\{\int_{\mathbb R}w^2\left|\mathbf S_{{\mathbf N}}[f]\left(w\right)\right|^2dw\right\}\ge \frac{B^2\mathcal C_\Psi}{4}\|f\|^4.
\end{eqnarray*}
\item  For the parametric set ${{\mathbf N}}=(\cos\theta,\sin\theta,-\sin\theta,\cos\theta;p,q),\quad \theta\ne n\pi$, Theorem \ref{thm hsb ucp} yields Heisenberg UCP  for fractional Stockwell transform \label{gst}
  \begin{eqnarray*}
\left\{\int_{\mathbb R}\int_{\mathbb R^+}b^2\left|[FrST^{{\mathbf N}}_\Psi f](a,b)\right|^2dadb\right\}\left\{\int_{\mathbb R}w^2\left|\mathbf S_{{\mathbf N}}[f]\left(w\right)\right|^2dw\right\}\ge \frac{\sin^2\theta\mathcal C_\Psi}{4}\|f\|^4.
\end{eqnarray*}
\item For the parametric set ${{\mathbf N}}=(0,1,-1,0;0,0),$ Theorem \ref{thm hsb ucp} yields Heisenberg UCP for the classical Stockwell transform (\ref{st2}).
    \begin{eqnarray*}
\left\{\int_{\mathbb R}\int_{\mathbb R^+}b^2\left|[ST^{{\mathbf N}}_\Psi f](a,b)\right|^2dadb\right\}\left\{\int_{\mathbb R}w^2\left|\mathbf S_{{\mathbf N}}[f]\left(w\right)\right|^2dw\right\}\ge \frac{\mathcal C_\Psi}{4}\|f\|^4.\\
\end{eqnarray*}
\end{itemize}
\end{remark}

In continuation, we shall derive the logarithmic UCP  for the proposed special affine Stockwell transform .

\begin{theorem}[Logarithmic UCP]\label{thm log ucp} Let $f(t)\in L^2(\mathbb R)$ be any signal and $\Psi\in L^2(\mathbb R)$ be the admissible window function, then the SAST of $f$  with respect to $\Psi$ denoted by
$[SAST^{{\mathbf N}}_\Psi f](a,b)$ satisfies the following uncertainty inequality holds:
\begin{eqnarray}
\nonumber&&\int_{\mathbb R}\int_{\mathbb R^+}\ln|b|\left|[SAST^{{\mathbf N}}_\Psi f](a,b)\right|^2dadb+\frac{\mathcal C_\Psi}{\sqrt{2\pi}}\int_{\mathbb R}\ln|w|\left|\mathbf S_{{\mathbf N}}[f]\left(w\right)\right|^2dw\\
&&\qquad\qquad\qquad\ge\frac{\mathcal C_\Psi}{2\pi}\left( \frac{\Gamma'(1/4)}{\Gamma(1/4)}+\ln|B|\right)\|f\|^2.
\end{eqnarray}
\end{theorem}
\begin{proof}
The logarithmic UCP for the SAFT reads \cite{cc1}
\begin{eqnarray}
\nonumber&&\int_{\mathbb R}\ln|t||f(t)|^2dt+\int_{\mathbb R}\ln|w||\mathbf S_{{\mathbf N}}[f](w)|^2dw\\
\label{1log qpft}&&\qquad\ge\left( \frac{\Gamma'(1/4)}{\Gamma(1/4)}+\ln|B|\right)\int_{\mathbb R}|f(t)|^2dt
\end{eqnarray}
Identifying $[SAST^{{\mathbf N}}_\Psi f](a,b)$ as a function of translation parameter $b$ and replacing $f(t)$ by $[SAST^{{\mathbf N}}_\Psi f](a,b)$ in (\ref{1hsb qpft}), we have
\begin{eqnarray}
\nonumber&&\int_{\mathbb R}\ln|b|\left|[SAST^{{\mathbf N}}_\Psi f](a,b)\right|^2db+\int_{\mathbb R}\ln|w|\left|\mathbf S_{{\mathbf N}}\left[[SAST^{{\mathbf N}}_\Psi f](a,b)\right](w)\right|^2dw\\
\label{2log qpft}&&\qquad\qquad\ge\left( \frac{\Gamma'(1/4)}{\Gamma(1/4)}+\ln|B|\right)\int_{\mathbb R}\left|[SAST^{{\mathbf N}}_\Psi f](a,b)\right|^2db
\end{eqnarray}
 On integrating (\ref{2log qpft}) with respect to $da$ , we have
\begin{eqnarray}
\nonumber&&\int_{\mathbb R^+}\left\{\int_{\mathbb R}\ln|b|\left|[SAST^{{\mathbf N}}_\Psi f](a,b)\right|^2db+\int_{\mathbb R}\ln|w|\left|\mathbf S_{{\mathbf N}}\left[[SAST^{{\mathbf N}}_\Psi f](a,b)\right](w)\right|^2dw\right\}da\\
\label{3log qpft}&&\qquad\qquad\ge\left( \frac{\Gamma'(1/4)}{\Gamma(1/4)}+\ln|B|\right)\int_{\mathbb R}\int_{\mathbb R^+}\left|[SAST^{{\mathbf N}}_\Psi f](a,b)\right|^2dadb
\end{eqnarray}
Implementing  Fubini's theorem to LHS and Corollary \ref{cor1} to RHS of (\ref{3log qpft}), it yields
\begin{eqnarray}
\nonumber&&\int_{\mathbb R}\int_{\mathbb R^+}\ln|b|\left|[SAST^{{\mathbf N}}_\Psi f](a,b)\right|^2dadb+\int_{\mathbb R}\int_{\mathbb R^+}\ln|w|\left|\mathbf S_{{\mathbf N}}\left[[SAST^{{\mathbf N}}_\Psi f](a,b)\right](w)\right|^2dadw\\
\label{4log qpft}&&\qquad\qquad\ge\left( \frac{\Gamma'(1/4)}{\Gamma(1/4)}+\ln|B|\right)\frac{\mathcal C_\Psi}{2\pi}\|f\|^2.
\end{eqnarray}
Using procedure of Eq.(\ref{5hsb qpft}) Theorem \ref{thm hsb ucp},  we have
\begin{eqnarray}\label{5log qpft}
\int_{\mathbb R}\int_{\mathbb R^+}\ln|w||\mathbf S_{{\mathbf N}}\left[SAST^{{\mathbf N}}_\Psi f](a,b)\right](w)|^2dadw&=&\frac{\mathcal C_\Psi}{\sqrt{2\pi}}\int_{\mathbb R}\ln|w|\left|\mathbf S_{{\mathbf N}}[f]\left(w\right)\right|^2dw.
\end{eqnarray}
Plugging (\ref{5log qpft}) in (\ref{4log qpft}), we get desired result as
\begin{eqnarray*}
\nonumber&&\int_{\mathbb R}\int_{\mathbb R^+}\ln|b|\left|[SAST^{{\mathbf N}}_\Psi f](a,b)\right|^2dadb+\frac{\mathcal C_\Psi}{\sqrt{2\pi}}\int_{\mathbb R}\ln|w|\left|\mathbf S_{{\mathbf N}}[f]\left(w\right)\right|^2dw\\
&&\qquad\qquad\qquad\ge\frac{\mathcal C_\Psi}{2\pi}\left( \frac{\Gamma'(1/4)}{\Gamma(1/4)}+\ln|B|\right)\|f\|^2.
\end{eqnarray*}

This completes the proof.
\end{proof}
Finally, we shall derive
Nazarov’s UCP which  measures the localization of a function $f$
by taking into consideration the notion of support of the function.
\begin{theorem}Let $[SAST^{{\mathbf N}}_\Psi f](a,b)$ be the SAST of  any signal $f\in L^2(\mathbb R)$ and $V,U$ be two measurable sets of $\mathbb R$. Then there exists a constant $K>0,$ such that
\begin{eqnarray}
\nonumber &&Ke^{K|V||U|}\left\{\int_{\mathbb R\setminus V}\int_{\mathbb R^+}\big|[SAST^{{\mathbf N}}_\Psi f](a,b)\big|^2dadb+\frac{\mathcal C_\Psi}{\sqrt{2\pi}}\int_{\mathbb R\setminus (UB)}|\mathbf S_{{\mathbf N}}[f](w)|^2dw\right\}\\
\label{ucp n4}&&\qquad\qquad\qquad\qquad\qquad\qquad\ge \frac{\mathcal C_\Psi}{2\pi}\|f\|^2,
\end{eqnarray}
where $|V|$ denotes Lebesgue measure measure of $V.$
\end{theorem}
\begin{proof}For any arbitrary function $f\in L^2(\mathbb R)$ and finite measurable sets $V$ and $U$ of $\mathbb R,$ Nazarov’s UCP in SAFT domain reads \cite{cc1}
\begin{equation}\label{ucp n1}
Ke^{K|V||U|}\left\{\int_{\mathbb R\setminus V}|f(t)|^2dt+\int_{\mathbb R\setminus (UB)}|\mathbf S_{{\mathbf N}}[f](w)|^2dw\right\}\ge\int_{\mathbb R}|f(t)|^2dt,
\end{equation}
where $|.|$ denotes Lebesgue measure and $K$ is positive constant.\\
By identifying $[SAST^{{\mathbf N}}_\Psi f](a,b)$ as a function of translation parameter $b$ followed by invoking Equation (\ref{ucp n1}),
we obtain
\begin{eqnarray}
\nonumber &&Ke^{K|V||U|}\left\{\int_{\mathbb R\setminus V}\big|[SAST^{{\mathbf N}}_\Psi f](a,b)\big|^2db+\int_{\mathbb R\setminus (UB)}|\mathbf S_{{\mathbf N}}\big[[SAST^{{\mathbf N}}_\Psi f](a,b)\big](w)|^2dw\right\}\\
\label{ucp n2}&&\qquad\qquad\qquad\ge\int_{\mathbb R}\big|[SAST^{{\mathbf N}}_\Psi f](a,b)\big|^2db.
\end{eqnarray}
Upon integrating Equation (\ref{ucp n2}) with respect to the $da$ and then implementing  Fubini's theorem to LHS, we have
\begin{eqnarray}
\nonumber &&Ke^{K|V||U|}\left\{\int_{\mathbb R\setminus V}\int_{\mathbb R^+}\big|[SAST^{{\mathbf N}}_\Psi f](a,b)\big|^2dadb+\int_{\mathbb R\setminus (UB)}\int_{\mathbb R^+}|\mathbf S_{{\mathbf N}}\big[[SAST^{{\mathbf N}}_\Psi f](a,b)\big](w)|^2dadw\right\}\\
\label{ucp n3}&&\qquad\qquad\qquad\ge\int_{\mathbb R}\int_{\mathbb R^+}\big|[SAST^{{\mathbf N}}_\Psi f](a,b)\big|^2dadb.
\end{eqnarray}
Finally, applying procedure of Eq.(\ref{5hsb qpft}) Theorem \ref{thm hsb ucp}, to LHS and and Corollary \ref{cor1} to RHS   we have
\begin{eqnarray*}
\nonumber &&Ke^{K|V||U|}\left\{\int_{\mathbb R\setminus V}\int_{\mathbb R^+}\big|[SAST^{{\mathbf N}}_\Psi f](a,b)\big|^2dadb+\frac{\mathcal C_\Psi}{\sqrt{2\pi}}\int_{\mathbb R\setminus (UB)}|\mathbf S_{{\mathbf N}}[f](w)|^2dw\right\}\\
\label{ucp n4}&&\qquad\qquad\qquad\qquad\qquad\qquad\ge \frac{\mathcal C_\Psi}{2\pi}\|f\|^2.
\end{eqnarray*}
This completes the proof of Theorem.
\end{proof}

\section{Potential Applications}\label{secpa}

In order to rectify the limitations of the SAFT and ST,  in this chapter we have proposed the    Stockwell transform, which not only preserves the properties of the classical Stockwell transform but also include a cluster of integral transforms as its special cases viz:  linear canonical Stockwell transform, fractional Stockwell transform and etc. The proposed special affine Stockwell transform have higher flexibility due to presence of  extra degrees of freedom and hence can be employed in optimizing the concentration of the higher frequency spectrum. The potential applications lies in the representation and detection of non-stationary LFM chirp signals, which are widely used in radar and sonar systems. In this  section we show that the proposed SAST is applied to detect
the LFM chirp signals: the mono-component, bi-component and echo interrupted  LFM signals. Moreover with the help of simulations we show that the proposed SAST exhibits better detection performance in comparison with the STFT, WT, WD and classical ST.

\subsection{ Detection of the mono-component and bi-component LFM signals  }

For illustration of detection
of LFM chirp signal, let us take one-component chirp signal as
\begin{equation}\label{eqn app}
f(t)= e^{i2\pi(\alpha t+\beta t^2)}
\end{equation}
where $\alpha$ and $\beta$ represent the initial frequency and frequency rate of $f(t)$, respectively.\\ For better time-frequency representation (TFR), we choose with variable coefficient: $\Psi_{\Delta_{gs}}(t)=1/\sqrt{2\pi\Delta_{gs}}e^{(-t^2/2\Delta^2_{gs})}.$ It is clear that the variation in the parameter $\Delta_{gs}$ affects  the time-frequency domain resolution,  larger values of $\Delta_{gs}$ decreases bandwidth and increases time width and hence improves the frequency domain resolution. Similarly, smaller values of $\Delta_{gs}$ increases bandwidth and decreases time width and hence improves the time domain resolution. From (\ref{eqn qpst2}), it is evident that the time-frequency performance of the proposed quadratic-phase Stockwell transform $[SAST^{{\mathbf N}}_\Psi f](v,\xi)$ depends on variable factor $\Delta_{gs}$ and the real parameter set ${\mathbf N}=(a,b,c,d;p,q)$. For instance, the TFR
 of one-component LFM chirp signal $f (t)$, at $\alpha=80$ and $\beta=30$ obtained by the STFT, WT, ST, scaled-WD and the proposed SAST are depicted in Fig.\ref{mc}, respectively. Fig.\ref{mc}(e)-(f) are obtained respectively by $[SAST^{{\mathbf N}}_\Psi f](v,\xi)$ with $\Delta_{gs}=12, {\bf N}=(12,5,4,0;0,0)$ and $\Delta_{gs}=20, A=(15,-1,4,0;0,0).$ Thus it is clear from Fig.\ref{mc} that with the help of above mentioned parameters the proposed SAST leads to a perfect time–frequency resolution for the one-component chirp signals.\\


 Now taking the bi-component chirp of the form
 \begin{equation}\label{eqn app2}
f(t)= \sum_{r=1}^2e^{i2\pi(\alpha_r t+\beta_r t^2)},
\end{equation}
the TFR
 of bi-component LFM chirp signal (\ref{eqn app2}), at $\alpha_1=80,\,\alpha_2=80$ and $\beta_1=27,\,\beta_2=30$ obtained by the STFT, WT, ST, scaled-WD and the proposed SAST are depicted in Fig.\ref{bc}, respectively. It is clear from Fig.\ref{bc} that the proposed SAST is free from cross-term and provides better resolution for the TFR of the bi-component chirp signal. Thus from Fig.\ref{mc} and Fig.\ref{bc}, we observe that that the proposed SAST has its application in detection of LFM chirp signals and has better resolution than the various existing classes of time-frequency analysis tools.\\

\subsection{ Detection of the echo interrupted LFM chirp signals signals  }
The potential application can also found in the detection of LFM chirp signals in the presence of echo signal which is also a non-stationary chirp signal. In the detection of   LFM chirp signals  containing echo the following procedure is adopted. Let us denote $S_{in}(t)=f_{in}(t)+n(t),$ be the chirp signal with echo [Note:$f_{in}(t)$ desired signal and $n(t)$ is echo ].  Firstly, the SAST  $[SAST^{{\mathbf N}}_\Psi ]$  is used as a transform technique, then the transformed input signal containing echo i.e, $[SAST^{{\mathbf N}}_\Psi S_{in}(t)](v,\xi)$ is multiplied by an appropriate filter impulse response say $[SAST^{{\mathbf N}}_\Psi F(t)](v,\xi)$. Finally the inverse transform (inverse SAST) is employed to
capture the desired echo-free output $S_{o}(t)=f_{out}(t)$.\\

The potential application can also found in the detection of disturbed chirp signals, interfered chirp signals and its usefulness can be seen in the time-frequency representation of the real-world bat echolocation signal \cite{st30,st31,st32,st60}. Moreover, it is clear from \cite{st60} the proposed transform can have potential application in different areas of Bio-medical, Geo-informatics,
Signal processing, Image processing,  Power system and Radar signal and Communication.
 \\In the flip side , the philosophy of the special affine Fourier transform Stockwell transform  UCPs is the same as that for the UCPs in the special affine Fourier transform  setting. The Heisenberg’s UCP for the special affine Fourier transform  states the relation of one signal in spatial and another  signal in
frequency domain. A potential application can be found in the estimation of the lower bound of integral \cite{cc1,[14]}. As such if a signal
in Theorem (\ref{thm hsb ucp}) is determined, it is difficult to obtain the value of the left side of Theorem (\ref{thm hsb ucp}) by calculating the integral. However,
an estimation of the left side of Theorem (\ref{thm hsb ucp}) can be easily obtained based on Theorem(\ref{thm hsb ucp}). Another  potential applications of  the
derived Theorem (\ref{thm hsb ucp}) can be used in the estimation of effective bandwidths \cite{aaa} in the SAST domain. For example if the spread of signal in the time domain $(T_\upsilon)$ is known then the bandwidth of a system that performs SAST cannot be narrower than $\frac{B^2\mathcal C_\Psi}{4T_\upsilon^2}.$
Nevertheless, the convolution structure exhibited
by the special affine Stockwell transform demonstrates that the proposed transform can also be employed
in filter design.

\section{Conclusion}\label{sec con}
 In the study, we introduce special affine Stockwell transform (SAST) by invoking  the convolution structure associated with SAFT, which rectifies the limitations of the
 quadratic-phase Fourier transform  and the Stockwell transform.
 Firstly, we examine the resolution of the SAST in the time and
SAFT domains and then derive some of its basic properties, such as Rayleigh’s energy theorem,
inversion formula and provide a characterization of the range. Besides, we also  derive a
direct relationship between the recently introduced Wigner-Ville distribution and the
proposed SAST. Moreover,
 we investigate Heisenberg,s uncertainty principle,  logarithmic uncertainty principle and Nazarov’s uncertainty principle associated with the proposed SAST.Finally, we study and investigate several applications based on SAST, including the detection of LFM  signals and echo interrupted chirp signals.

\section*{Declarations}
\begin{itemize}
\item  Availability of data and materials: The data is provided on the request to the authors.
\item Competing interests: The authors have no competing interests.
\item Funding: No funding was received for this work
\item Author's contribution: Both the authors equally contributed towards this work.
\item Acknowledgements: This work is supported by the  Research  Grant\\
(No. JKST\&IC/SRE/J/357-60) provided by JKSTIC, Govt. of Jammu and Kashmir,
India.

\end{itemize}

{\bf{References}}

\end{document}